\pdfoutput=1
\documentclass[sigconf]{acmart}

\usepackage{hyperref}

\usepackage[ruled, vlined, nofillcomment, linesnumbered]{algorithm2e}
\usepackage{url}
\usepackage{amsmath}
\usepackage{graphicx}
\usepackage{subcaption}
\usepackage[english]{babel}
\usepackage{booktabs}
\usepackage{verbatim}
\usepackage{float}
\usepackage{footmisc}
\usepackage{xspace}
\usepackage{tikz}
\usepackage{pgfplots}
\usepackage{pgfplotstable}
\usetikzlibrary{shapes.geometric}
\usepackage{balance} 
\usepackage{etoolbox}

\makeatletter

\newcommand{\set}[1]{\left\{#1\right\}}
\newcommand{\pr}[1]{\left(#1\right)}
\newcommand{\fpr}[1]{\mathopen{}\left(#1\right)}

\newcommand{\abs}[1]{{\left|#1\right|}}

\newcommand{\p}{\textbf{P}}
\newcommand{\np}{\textbf{NP}}

\newcommand{\define}{\leftarrow}

\newcommand{\ceil}[1]{\left\lceil#1\right\rceil}

\DeclareRobustCommand{\dispfunc}[2]{%
  \ensuremath{%
  \ifthenelse{\equal{#2}{}}%
    {\mathit{#1}}%
    {\mathit{#1}\fpr{#2}}}}

\newcommand{\score}[1]{\dispfunc{q}{#1}}
\newcommand{\scorev}[1]{\dispfunc{\deg_{\lambda}}{#1}}

\newcommand{\bigO}[1]{\dispfunc{\mathcal{O}}{#1}}

\newcommand{\alggreedy}{\textsc{Grd-Naive}\xspace}

\newcommand{\alggreedyfastest}{\textsc{Greedy}\xspace}

\newcommand{\algdenseg}{\textsc{STC-Cut}\xspace}
\newcommand{\algdensec}{\textsc{STC-Peel}\xspace}
\newcommand{\alglpstc}{\textsc{STC-LP}\xspace}
\newcommand{\algip}{\textsc{STC-ILP}\xspace}
\newcommand{\algdg}{\textsc{Cut}\xspace}
\newcommand{\algdc}{\textsc{Peel}\xspace}

\newcommand{\algvcm}{\textsc{VC-Mat}\xspace}
\newcommand{\algminstc}{\textsc{Apr-MinSTC}\xspace}

\newcommand{\algdynvc}{\textsc{Dynamic-Vertex-cover}\xspace}

\newcommand{\acu}{\textsc{CT}\xspace}
\newcommand{\ape}{\textsc{PL}\xspace}
\newcommand{\ago}{\textsc{GR}\xspace}
\newcommand{\agr}{\textsc{GR}\xspace}
\newcommand{\aip}{\textsc{IP}\xspace}
\newcommand{\alp}{\textsc{LP}\xspace}

\newcommand{\dtname}[1]{\textsl{#1}}

\newcommand{\calS}{\ensuremath{{\mathcal S}}}

\newcommand{\prbcovermin}{\textsc{Min-Vertex-cover}\xspace}

\newcommand{\prbmaxclique}{\textsc{Max-Clique}\xspace}

\newcommand{\prbminSTC}{\textsc{MinSTC}\xspace}

\newcommand{\prbstrwk}{\textsc{stc-den}\xspace}
\newcommand{\prbstrwkr}{\textsc{stc-relax}\xspace}

\newcommand{\stc}{STC\xspace}

\newtheorem{problem}{Problem}

\SetKwComment{tcpas}{\{}{\}}
\SetCommentSty{textnormal}
\SetArgSty{textnormal}
\SetKw{False}{false}
\SetKw{True}{true}
\SetKw{Null}{null}
\SetKwInOut{Output}{output}
\SetKwInOut{Input}{input}
\SetKw{AND}{and}
\SetKw{OR}{or}
\SetKw{Break}{break}

\SetKwFor{Until}{until}{do}{}

\definecolor{yafcolor5}{rgb}{0.141, 0.345, 0.643}
\SetKw{Output}{output}

\DeclareRobustCommand{\dispfunc}[2]{%
	\ensuremath{%
		\ifthenelse{\equal{#2}{}}%
			{\mathit{#1}}%
			{\mathit{#1}\fpr{#2}}}}

\SetKw{Break}{break}

\pgfdeclarelayer{background}
\pgfdeclarelayer{foreground}
\pgfsetlayers{background,main,foreground}

\definecolor{yafaxiscolor}{rgb}{0.3, 0.3, 0.3}

\definecolor{yafcolor1}{rgb}{0.4, 0.165, 0.553}
\definecolor{yafcolor2}{rgb}{0.949, 0.482, 0.216}
\definecolor{yafcolor3}{rgb}{0.47, 0.549, 0.306}
\definecolor{yafcolor4}{rgb}{0.925, 0.165, 0.224}
\definecolor{yafcolor5}{rgb}{0.141, 0.345, 0.643}
\definecolor{yafcolor6}{rgb}{0.965, 0.933, 0.267}
\definecolor{yafcolor7}{rgb}{0.627, 0.118, 0.165}
\definecolor{yafcolor8}{rgb}{0.878, 0.475, 0.686}

\newlength{\yafaxispad}
\newlength{\yaftlpad}
\newlength{\yaflabelpad}
\newlength{\yafaxiswidth}
\newlength{\yafticklen}
\makeatletter
\def\pgfplots@drawtickgridlines@INSTALLCLIP@onorientedsurf#1{}
\makeatother

\newcommand{\yafdrawaxis}[4]{
	\pgfplotstransformcoordinatex{#1}\let\xmincoord=\pgfmathresult 
	\pgfplotstransformcoordinatex{#2}\let\xmaxcoord=\pgfmathresult 
	\pgfplotstransformcoordinatey{#3}\let\ymincoord=\pgfmathresult 
	\pgfplotstransformcoordinatey{#4}\let\ymaxcoord=\pgfmathresult 
	\pgfsetlinewidth{\yafaxiswidth} 
	\pgfsetcolor{yafaxiscolor}
	\pgfpathmoveto{\pgfpointadd{\pgfpointadd{\pgfplotspointrelaxisxy{0}{0}}{\pgfqpointxy{\xmincoord}{0}}}{\pgfqpoint{-0.5\yafaxiswidth}{\yafaxispad}}}
	\pgfpathlineto{\pgfpointadd{\pgfpointadd{\pgfplotspointrelaxisxy{0}{0}}{\pgfqpointxy{\xmaxcoord}{0}}}{\pgfqpoint{0.5\yafaxiswidth}{\yafaxispad}}}
	\pgfpathmoveto{\pgfpointadd{\pgfpointadd{\pgfplotspointrelaxisxy{0}{0}}{\pgfqpointxy{0}{\ymincoord}}}{\pgfqpoint{\yafaxispad}{-0.5\yafaxiswidth}}}
	\pgfpathlineto{\pgfpointadd{\pgfpointadd{\pgfplotspointrelaxisxy{0}{0}}{\pgfqpointxy{0}{\ymaxcoord}}}{\pgfqpoint{\yafaxispad}{0.5\yafaxiswidth}}}
	\pgfusepath{stroke}
}

\pgfplotscreateplotcyclelist{yaf}{%
{yafcolor5,mark options={scale=0.75},mark=o}, 
{yafcolor2,mark options={scale=0.75},mark=square},
{yafcolor3,mark options={scale=0.75},mark=triangle},
{yafcolor4,mark options={scale=0.75},mark=o},
{yafcolor1,mark options={scale=0.75},mark=o},
{yafcolor6,mark options={scale=0.75},mark=o},
{yafcolor7,mark options={scale=0.75},mark=o},
{yafcolor8,mark options={scale=0.75},mark=o}} 

\pgfkeys{/pgf/number format/.cd,1000 sep={\,}}

\pgfplotsset{axis y line=left, axis x line=bottom,
	tick align=outside,
	tickwidth=\yafticklen,
	clip = false,
    x axis line style= {-, line width = 0pt, color=black!0},
    y axis line style= {-, line width = 0pt, color=black!0},
    x tick style= {line width = \yafaxiswidth, color=yafaxiscolor, yshift = \yafaxispad},
    y tick style= {line width = \yafaxiswidth, color=yafaxiscolor, xshift = \yafaxispad},
    x tick label style = {font=\small, yshift = \yaftlpad, inner xsep = 0pt},
    y tick label style = {font=\small, xshift = \yaftlpad},
    every axis y label/.style = {at = {(ticklabel cs:0.5)}, rotate=90, anchor=center, font=\small, yshift = -\yaflabelpad, inner sep = 0pt},
    every axis x label/.style = {at = {(ticklabel cs:0.5)}, anchor=center, font=\small, yshift = \yaflabelpad},
    x tick label style = {font=\small, yshift = 1pt},
    grid = major,
    major grid style  = {dash pattern = on 1pt off 3 pt},
	every axis plot post/.append style= {line width=\yafaxiswidth} ,
	legend cell align = left,
	legend style = {inner sep = 0pt, cells = {font=\scriptsize}},
	legend image code/.code={%
		\draw[mark repeat=2,mark phase=2,#1] 
		plot coordinates { (0cm,0cm) (0.15cm,0cm) (0.3cm,0cm) };%
	} 
}

\newcommand{\pgfprintduration}[1]{%
	\ifthenelse{\equal{#1}{}}{---}{%
	\pgfmathsetmacro{\minutes}{floor(#1 / 60)}%
	\pgfmathsetmacro{\seconds}{#1 - 60*\minutes}%
	\pgfmathifthenelse{\minutes > 0}{"\pgfmathprintnumber{\minutes}m \pgfmathprintnumber[fixed,precision=0]{\seconds}s"}{"\pgfmathprintnumber{\seconds}s"}\pgfmathresult}}

\copyrightyear{2024}
\acmYear{2024}
\setcopyright{rightsretained}
\acmConference[KDD '24]{Proceedings of the 30th ACM SIGKDD Conference on Knowledge Discovery and Data Mining}{August 25--29, 2024}{Barcelona, Spain} 
\acmBooktitle{Proceedings of the 30th ACM SIGKDD Conference on Knowledge Discovery and Data Mining (KDD '24), August 25--29, 2024, Barcelona, Spain}
\acmDOI{10.1145/3637528.3671697} 
\acmISBN{979-8-4007-0490-1/24/08}

\makeatletter
 \gdef
\makeatother

\begin{document}

\author{Chamalee Wickrama Arachchi}
\affiliation{%
  \institution{University of Helsinki}
  \city{Helsinki}
  \country{Finland}
}
\email{chamalee.wickramaarachch@helsinki.fi}

\author{Iiro Kumpulainen}
\affiliation{%
  \institution{University of Helsinki}
  \city{Helsinki}
  \country{Finland}
}
\email{iiro.kumpulainen@helsinki.fi}

\author{Nikolaj Tatti}
\affiliation{%
  \institution{HIIT, University of Helsinki}
  \city{Helsinki}
  \country{Finland}
}
\email{nikolaj.tatti@helsinki.fi}

\date{}

\title{Dense Subgraph Discovery Meets Strong Triadic Closure}

\begin{abstract}
Finding dense subgraphs is a core problem with numerous graph mining applications such as community detection in social networks and anomaly detection. However, in many real-world networks connections are not equal. One way to label edges as either strong or weak is to use strong triadic closure~(\stc). Here, if one node connects strongly with two other nodes, then those two nodes should be connected at least with a weak edge. \stc-labelings are not unique and finding the maximum number of strong edges is \np-hard. In this paper, we apply \stc to dense subgraph discovery. More formally, our score for a given subgraph is the ratio between the sum of the number of strong edges and weak edges, weighted by a user parameter $\lambda$, and the number of nodes of the subgraph. Our goal is to find a subgraph and an \stc-labeling maximizing the score.
We show that for $\lambda = 1$, our problem is equivalent to finding the densest subgraph, while for $\lambda = 0$, our problem is equivalent to finding the largest clique, making our problem \np-hard.
We propose an exact algorithm based on integer linear programming and four practical polynomial-time heuristics. We present an extensive experimental study that shows that our algorithms can find the ground truth in synthetic datasets and run efficiently in real-world datasets.
\end{abstract}

\ccsdesc[500]{Theory of computation~Graph algorithms analysis}

\keywords{dense subgraph, strong triadic closure, integer linear programming}

\maketitle

\section{Introduction}

Many social networks naturally contain both {\em strongly} connected and {\em weakly} connected interactions among the entities of the network.
A question of particular interest is that given a set of pairwise user interactions, how to infer the strength of the social ties within the network? In other words, how to label the edges of an undirected graph as either strong or weak?

The problem of inferring the strength of social ties based on {\em strong triadic closure} principle~(\stc) has drawn attention over the past decade within the data mining community~\cite{sintos2014using,rozenshtein2017inferring,adriaens2020relaxing,oettershagen2022inferring,konstantinidis2018strong,matakos2022strengthening}. The \stc property assumes that there exist two types of social ties in the network: either {\em strong} or {\em weak}.
Let $A, B$, and $C$ be three entities in the network.
If the entities $A$ and $B$ are strongly connected with the entity $C$, then there should be at least a weak connection between $A$ and $B$. In other words, if both $A$ and $B$ are strong friends of $C$, then some kind of connection between $A$ and $B$ should also exist.
Note that these labels are not known and the goal is to infer the labels from the given unlabeled graph.

In this paper, we incorporate the \stc property into the problem of dense subgraph discovery~\cite{sintos2014using, goldberg1984finding}. More formally, given a subgraph and a weight parameter $\lambda$, we define a score as the ratio between the sum of the number of {strong} and {weak} edges weighted by $\lambda$ and the number of nodes within the subgraph.
Our objective is to find a subgraph \emph{and} a labeling that maximize our score while satisfying the \stc property within the subgraph. 

We will see that when $\lambda = 0$ finding an optimal subgraph is equal to finding a maximum clique. On the other hand, for $\lambda = 1$, finding an optimal subgraph is equal to finding the densest subgraph, that is, a subgraph $U$ maximizing the ratio of edges and nodes, $\abs{E(U)}/\abs{U}$. Both of these problems are well-studied. Optimizing the score for $0 < \lambda < 1$, yields a problem that is between these two cases. We expect that for small $\lambda$s the returned subgraph resembles a clique whereas large $\lambda$s yield a subgraph similar to the densest subgraph.

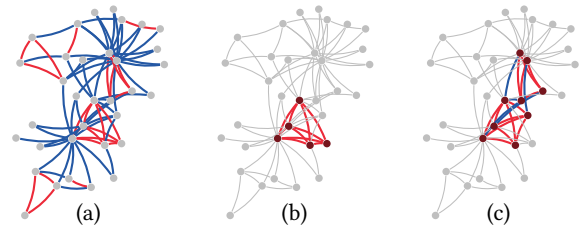
\begin{figure}
\begin{center}
\begin{tikzpicture}[scale = .6]
\tikzstyle{lm1} = [fill = yafcolor5!50!black, inner sep = 1.2pt]
\tikzstyle{lm2} = [circle, fill = yafcolor4!50!black, inner sep = 1pt]
\tikzstyle{lm3} = [circle, fill = gray!50, inner sep = 1pt]
\tikzstyle{ll} = []
\tikzstyle{le1} = [yafcolor5, thick, bend left = 10]
\tikzstyle{le2} = [yafcolor4, thick, bend left = 10]
\tikzstyle{le3} = [gray!50, bend left = 10]
\node[lm3] (n0) at (-0.39521, -0.90369) {};
\node[lm3] (n1) at (-0.1494, -0.63131) {};
\node[lm3] (n2) at (0.089756, -0.058231) {};
\node[lm3] (n3) at (0.32331, -1.0532) {};
\node[lm3] (n4) at (0.023263, -1.9859) {};
\node[lm3] (n5) at (-1.2108, -1.6373) {};
\node[lm3] (n6) at (-0.73841, -1.9544) {};
\node[lm3] (n7) at (0.69671, -1.0172) {};
\node[lm3] (n8) at (0.45414, -0.064462) {};
\node[lm3] (n9) at (1.35, 0.053898) {};
\node[lm3] (n10) at (-0.47034, -1.7978) {};
\node[lm3] (n11) at (-1.6432, -0.89285) {};
\node[lm3] (n12) at (0.52145, -1.7831) {};
\node[lm3] (n13) at (0.59789, -0.39239) {};
\node[lm3] (n14) at (1.15, 1.6278) {};
\node[lm3] (n15) at (1.644, 0.73643) {};
\node[lm3] (n16) at (-1.4464, -2.6101) {};
\node[lm3] (n17) at (-1.278, -0.60231) {};
\node[lm3] (n18) at (0.51654, 1.9636) {};
\node[lm3] (n19) at (-0.30691, 0.035621) {};
\node[lm3] (n20) at (0.10109, 1.891) {};
\node[lm3] (n21) at (-1.1086, -1.0869) {};
\node[lm3] (n22) at (1.4907, 1.075) {};
\node[lm3] (n23) at (-0.2767, 1.6423) {};
\node[lm3] (n24) at (-1.5541, 0.77017) {};
\node[lm3] (n25) at (-1.2836, 1.3377) {};
\node[lm3] (n26) at (1.5531, 1.5267) {};
\node[lm3] (n27) at (-0.5343, 0.9102) {};
\node[lm3] (n28) at (-0.16321, 0.83704) {};
\node[lm3] (n29) at (0.70218, 1.7305) {};
\node[lm3] (n30) at (0.93024, 0.14795) {};
\node[lm3] (n31) at (-0.5944, 0.37166) {};
\node[lm3] (n32) at (0.42538, 0.99008) {};
\node[lm3] (n33) at (0.58393, 0.82341) {};
\begin{pgfonlayer}{background}
\draw (n0) edge[le2] (n1);
\draw (n0) edge[le2] (n2);
\draw (n0) edge[le2] (n3);
\draw (n0) edge[le1] (n4);
\draw (n0) edge[le1] (n5);
\draw (n0) edge[le1] (n6);
\draw (n0) edge[le2] (n7);
\draw (n0) edge[le1] (n8);
\draw (n0) edge[le1] (n10);
\draw (n0) edge[le1] (n11);
\draw (n0) edge[le1] (n12);
\draw (n0) edge[le1] (n13);
\draw (n0) edge[le1] (n17);
\draw (n0) edge[le1] (n19);
\draw (n0) edge[le1] (n21);
\draw (n0) edge[le1] (n31);
\draw (n1) edge[le2] (n2);
\draw (n1) edge[le2] (n3);
\draw (n1) edge[le2] (n7);
\draw (n1) edge[le1] (n13);
\draw (n1) edge[le1] (n17);
\draw (n1) edge[le1] (n19);
\draw (n1) edge[le1] (n21);
\draw (n1) edge[le1] (n30);
\draw (n2) edge[le2] (n3);
\draw (n2) edge[le2] (n7);
\draw (n2) edge[le1] (n8);
\draw (n2) edge[le1] (n9);
\draw (n2) edge[le1] (n13);
\draw (n2) edge[le1] (n27);
\draw (n2) edge[le1] (n28);
\draw (n2) edge[le1] (n32);
\draw (n3) edge[le2] (n7);
\draw (n3) edge[le1] (n12);
\draw (n3) edge[le1] (n13);
\draw (n4) edge[le1] (n6);
\draw (n4) edge[le2] (n10);
\draw (n5) edge[le2] (n6);
\draw (n5) edge[le1] (n10);
\draw (n5) edge[le2] (n16);
\draw (n6) edge[le2] (n16);
\draw (n8) edge[le2] (n30);
\draw (n8) edge[le2] (n32);
\draw (n8) edge[le2] (n33);
\draw (n9) edge[le1] (n33);
\draw (n13) edge[le1] (n33);
\draw (n14) edge[le1] (n32);
\draw (n14) edge[le1] (n33);
\draw (n15) edge[le1] (n32);
\draw (n15) edge[le1] (n33);
\draw (n18) edge[le1] (n32);
\draw (n18) edge[le1] (n33);
\draw (n19) edge[le1] (n33);
\draw (n20) edge[le1] (n32);
\draw (n20) edge[le1] (n33);
\draw (n22) edge[le1] (n32);
\draw (n22) edge[le1] (n33);
\draw (n23) edge[le1] (n25);
\draw (n23) edge[le2] (n27);
\draw (n23) edge[le1] (n29);
\draw (n23) edge[le1] (n32);
\draw (n23) edge[le1] (n33);
\draw (n24) edge[le2] (n25);
\draw (n24) edge[le1] (n27);
\draw (n24) edge[le2] (n31);
\draw (n25) edge[le2] (n31);
\draw (n26) edge[le2] (n29);
\draw (n26) edge[le1] (n33);
\draw (n27) edge[le1] (n33);
\draw (n28) edge[le1] (n31);
\draw (n28) edge[le1] (n33);
\draw (n29) edge[le1] (n32);
\draw (n29) edge[le1] (n33);
\draw (n30) edge[le2] (n32);
\draw (n30) edge[le2] (n33);
\draw (n31) edge[le1] (n32);
\draw (n31) edge[le1] (n33);
\draw (n32) edge[le2] (n33);
\end{pgfonlayer}

\node[anchor=south west] at (-0.5, -3) {(a)};
\end{tikzpicture}
\hspace{0.5cm}
\begin{tikzpicture}[scale = .6]
\tikzstyle{lm1} = [fill = yafcolor3!50!black, inner sep = 1.2pt]
\tikzstyle{lm2} = [circle, fill = yafcolor4!50!black, inner sep = 1pt]
\tikzstyle{lm3} = [circle, fill = gray!50, inner sep = 1pt]
\tikzstyle{ll} = []
\tikzstyle{le1} = [yafcolor4, thick, bend left = 10]
\tikzstyle{le2} = [yafcolor5, thick, bend left = 10]
\tikzstyle{le3} = [gray!50, bend left = 10]
\node[lm2] (n0) at (-0.39521, -0.90369) {};
\node[lm2] (n1) at (-0.1494, -0.63131) {};
\node[lm2] (n2) at (0.089756, -0.058231) {};
\node[lm2] (n3) at (0.32331, -1.0532) {};
\node[lm3] (n4) at (0.023263, -1.9859) {};
\node[lm3] (n5) at (-1.2108, -1.6373) {};
\node[lm3] (n6) at (-0.73841, -1.9544) {};
\node[lm2] (n7) at (0.69671, -1.0172) {};
\node[lm3] (n8) at (0.45414, -0.064462) {};
\node[lm3] (n9) at (1.35, 0.053898) {};
\node[lm3] (n10) at (-0.47034, -1.7978) {};
\node[lm3] (n11) at (-1.6432, -0.89285) {};
\node[lm3] (n12) at (0.52145, -1.7831) {};
\node[lm3] (n13) at (0.59789, -0.39239) {};
\node[lm3] (n14) at (1.15, 1.6278) {};
\node[lm3] (n15) at (1.644, 0.73643) {};
\node[lm3] (n16) at (-1.4464, -2.6101) {};
\node[lm3] (n17) at (-1.278, -0.60231) {};
\node[lm3] (n18) at (0.51654, 1.9636) {};
\node[lm3] (n19) at (-0.30691, 0.035621) {};
\node[lm3] (n20) at (0.10109, 1.891) {};
\node[lm3] (n21) at (-1.1086, -1.0869) {};
\node[lm3] (n22) at (1.4907, 1.075) {};
\node[lm3] (n23) at (-0.2767, 1.6423) {};
\node[lm3] (n24) at (-1.5541, 0.77017) {};
\node[lm3] (n25) at (-1.2836, 1.3377) {};
\node[lm3] (n26) at (1.5531, 1.5267) {};
\node[lm3] (n27) at (-0.5343, 0.9102) {};
\node[lm3] (n28) at (-0.16321, 0.83704) {};
\node[lm3] (n29) at (0.70218, 1.7305) {};
\node[lm3] (n30) at (0.93024, 0.14795) {};
\node[lm3] (n31) at (-0.5944, 0.37166) {};
\node[lm3] (n32) at (0.42538, 0.99008) {};
\node[lm3] (n33) at (0.58393, 0.82341) {};
\begin{pgfonlayer}{background}

\draw (n0) edge[le1] (n1);
\draw (n0) edge[le1] (n2);
\draw (n0) edge[le1] (n3);

\draw (n1) edge[le1] (n2);
\draw (n1) edge[le1] (n3);
\draw (n2) edge[le1] (n3);

\draw (n0) edge[le1] (n7);
\draw (n1) edge[le1] (n7);
\draw (n2) edge[le1] (n7);

\draw (n3) edge[le1] (n7);

\draw (n0) edge[le3] (n13);
\draw (n2) edge[le3] (n13);
\draw (n1) edge[le3] (n13);
\draw (n3) edge[le3] (n13);
\draw (n0) edge[le3] (n4);
\draw (n0) edge[le3] (n5);
\draw (n0) edge[le3] (n6);

\draw (n0) edge[le3] (n8);
\draw (n0) edge[le3] (n10);
\draw (n0) edge[le3] (n11);
\draw (n0) edge[le3] (n12);
\draw (n0) edge[le3] (n17);
\draw (n0) edge[le3] (n19);
\draw (n0) edge[le3] (n21);
\draw (n0) edge[le3] (n31);

\draw (n1) edge[le3] (n17);
\draw (n1) edge[le3] (n19);
\draw (n1) edge[le3] (n21);
\draw (n1) edge[le3] (n30);

\draw (n2) edge[le3] (n8);
\draw (n2) edge[le3] (n9);
\draw (n2) edge[le3] (n27);
\draw (n2) edge[le3] (n28);
\draw (n2) edge[le3] (n32);

\draw (n3) edge[le3] (n12);
\draw (n4) edge[le3] (n6);
\draw (n4) edge[le3] (n10);
\draw (n5) edge[le3] (n6);
\draw (n5) edge[le3] (n10);
\draw (n5) edge[le3] (n16);
\draw (n6) edge[le3] (n16);
\draw (n8) edge[le3] (n30);
\draw (n8) edge[le3] (n32);
\draw (n8) edge[le3] (n33);
\draw (n9) edge[le3] (n33);
\draw (n13) edge[le3] (n33);
\draw (n14) edge[le3] (n32);
\draw (n14) edge[le3] (n33);
\draw (n15) edge[le3] (n32);
\draw (n15) edge[le3] (n33);
\draw (n18) edge[le3] (n32);
\draw (n18) edge[le3] (n33);
\draw (n19) edge[le3] (n33);
\draw (n20) edge[le3] (n32);
\draw (n20) edge[le3] (n33);
\draw (n22) edge[le3] (n32);
\draw (n22) edge[le3] (n33);
\draw (n23) edge[le3] (n25);
\draw (n23) edge[le3] (n27);
\draw (n23) edge[le3] (n29);
\draw (n23) edge[le3] (n32);
\draw (n23) edge[le3] (n33);
\draw (n24) edge[le3] (n25);
\draw (n24) edge[le3] (n27);
\draw (n24) edge[le3] (n31);
\draw (n25) edge[le3] (n31);
\draw (n26) edge[le3] (n29);
\draw (n26) edge[le3] (n33);
\draw (n27) edge[le3] (n33);
\draw (n28) edge[le3] (n31);
\draw (n28) edge[le3] (n33);
\draw (n29) edge[le3] (n32);
\draw (n29) edge[le3] (n33);
\draw (n30) edge[le3] (n32);
\draw (n30) edge[le3] (n33);
\draw (n31) edge[le3] (n32);
\draw (n31) edge[le3] (n33);
\draw (n32) edge[le3] (n33);
\end{pgfonlayer}


\node[anchor=south west] at (-0.5, -3) {(b)};
\end{tikzpicture}
\hspace{0.5cm}
\begin{tikzpicture}[scale = .6]
\tikzstyle{lm1} = [fill = yafcolor3!50!black, inner sep = 1.2pt]
\tikzstyle{lm2} = [circle, fill = yafcolor4!50!black, inner sep = 1pt]
\tikzstyle{lm3} = [circle, fill = gray!50, inner sep = 1pt]
\tikzstyle{ll} = []
\tikzstyle{le1} = [yafcolor4, thick, bend left = 10]
\tikzstyle{le2} = [yafcolor5, thick, bend left = 10]
\tikzstyle{le3} = [gray!50, bend left = 10]
\node[lm2] (n0) at (-0.39521, -0.90369) {};
\node[lm2] (n1) at (-0.1494, -0.63131) {};
\node[lm2] (n2) at (0.089756, -0.058231) {};
\node[lm2] (n3) at (0.32331, -1.0532) {};
\node[lm3] (n4) at (0.023263, -1.9859) {};
\node[lm3] (n5) at (-1.2108, -1.6373) {};
\node[lm3] (n6) at (-0.73841, -1.9544) {};
\node[lm3] (n7) at (0.69671, -1.0172) {};
\node[lm2] (n8) at (0.45414, -0.064462) {};
\node[lm3] (n9) at (1.35, 0.053898) {};
\node[lm3] (n10) at (-0.47034, -1.7978) {};
\node[lm3] (n11) at (-1.6432, -0.89285) {};
\node[lm3] (n12) at (0.52145, -1.7831) {};
\node[lm2] (n13) at (0.59789, -0.39239) {};
\node[lm3] (n14) at (1.15, 1.6278) {};
\node[lm3] (n15) at (1.644, 0.73643) {};
\node[lm3] (n16) at (-1.4464, -2.6101) {};
\node[lm3] (n17) at (-1.278, -0.60231) {};
\node[lm3] (n18) at (0.51654, 1.9636) {};
\node[lm3] (n19) at (-0.30691, 0.035621) {};
\node[lm3] (n20) at (0.10109, 1.891) {};
\node[lm3] (n21) at (-1.1086, -1.0869) {};
\node[lm3] (n22) at (1.4907, 1.075) {};
\node[lm3] (n23) at (-0.2767, 1.6423) {};
\node[lm3] (n24) at (-1.5541, 0.77017) {};
\node[lm3] (n25) at (-1.2836, 1.3377) {};
\node[lm3] (n26) at (1.5531, 1.5267) {};
\node[lm3] (n27) at (-0.5343, 0.9102) {};
\node[lm3] (n28) at (-0.16321, 0.83704) {};
\node[lm3] (n29) at (0.70218, 1.7305) {};
\node[lm2] (n30) at (0.93024, 0.14795) {};
\node[lm3] (n31) at (-0.5944, 0.37166) {};
\node[lm2] (n32) at (0.42538, 0.99008) {};
\node[lm2] (n33) at (0.58393, 0.82341) {};
\begin{pgfonlayer}{background}


\draw (n0) edge[le1] (n1);
\draw (n0) edge[le1] (n2);
\draw (n0) edge[le1] (n3);
\draw (n0) edge[le1] (n13);

\draw (n1) edge[le1] (n2);
\draw (n1) edge[le1] (n3);
\draw (n1) edge[le1] (n13);

\draw (n2) edge[le1] (n3);
\draw (n2) edge[le1] (n13);

\draw (n3) edge[le1] (n13);

\draw (n8) edge[le1] (n32);
\draw (n8) edge[le1] (n30);
\draw (n8) edge[le1] (n33);

\draw (n30) edge[le1] (n33);
\draw (n30) edge[le1] (n32);
\draw (n32) edge[le1] (n33);


\draw (n0) edge[le2] (n8);
\draw (n1) edge[le2] (n30);
\draw (n2) edge[le2] (n8);
\draw (n2) edge[le2] (n32);
\draw (n13) edge[le2] (n33);


\draw (n0) edge[le3] (n31);
\draw (n2) edge[le3] (n28);
\draw (n28) edge[le3] (n33);
\draw (n31) edge[le3] (n32);
\draw (n31) edge[le3] (n33);
\draw (n28) edge[le3] (n31);

\draw (n0) edge[le3] (n4);
\draw (n0) edge[le3] (n5);
\draw (n0) edge[le3] (n6);
\draw (n0) edge[le3] (n7);

\draw (n0) edge[le3] (n10);
\draw (n0) edge[le3] (n11);
\draw (n0) edge[le3] (n12);
\draw (n0) edge[le3] (n17);
\draw (n0) edge[le3] (n19);
\draw (n0) edge[le3] (n21);
\draw (n1) edge[le3] (n7);
\draw (n1) edge[le3] (n17);
\draw (n1) edge[le3] (n19);
\draw (n1) edge[le3] (n21);

\draw (n2) edge[le3] (n7);

\draw (n2) edge[le3] (n9);

\draw (n2) edge[le3] (n27);

\draw (n3) edge[le3] (n7);
\draw (n3) edge[le3] (n12);
\draw (n4) edge[le3] (n6);
\draw (n4) edge[le3] (n10);
\draw (n5) edge[le3] (n6);
\draw (n5) edge[le3] (n10);
\draw (n5) edge[le3] (n16);
\draw (n6) edge[le3] (n16);

\draw (n9) edge[le3] (n33);

\draw (n14) edge[le3] (n32);
\draw (n14) edge[le3] (n33);
\draw (n15) edge[le3] (n32);
\draw (n15) edge[le3] (n33);
\draw (n18) edge[le3] (n32);
\draw (n18) edge[le3] (n33);
\draw (n19) edge[le3] (n33);
\draw (n20) edge[le3] (n32);
\draw (n20) edge[le3] (n33);
\draw (n22) edge[le3] (n32);
\draw (n22) edge[le3] (n33);
\draw (n23) edge[le3] (n25);
\draw (n23) edge[le3] (n27);
\draw (n23) edge[le3] (n29);
\draw (n23) edge[le3] (n32);
\draw (n23) edge[le3] (n33);
\draw (n24) edge[le3] (n25);
\draw (n24) edge[le3] (n27);
\draw (n24) edge[le3] (n31);
\draw (n25) edge[le3] (n31);
\draw (n26) edge[le3] (n29);
\draw (n26) edge[le3] (n33);
\draw (n27) edge[le3] (n33);

\draw (n29) edge[le3] (n32);
\draw (n29) edge[le3] (n33);

\end{pgfonlayer}











\node[anchor=south west] at (-0.5, -3) {(c)};
\end{tikzpicture}

\caption{
Strong~(Red) and weak~(Blue) edges of the Karate club dataset maximizing the number of strong edges (a), $\lambda = 0$ (b), and $\lambda = 0.5$ (c) using our integer linear program based algorithm~(\algip). 
We define our score as the sum of the number of strong and weak edges weighted by a parameter $\lambda$, divided by the size of the subgraph.
The scores are $2.0$ and $2.06$ for (b) and (c), respectively. We see that (b) is a clique of size $5$.}
\label{fig:karate}
\end{center}
\end{figure}

\textbf{Example:} 
We illustrate the difference between our problem and the original STC problem considered by \citet{sintos2014using} in Figure~\ref{fig:karate}.
The goal of this paper is to find a subgraph that maximizes our score while satisfying the \stc property. In contrast, \citet{sintos2014using} aims to label \emph{all} the edges in the graph such that the number of weak edges is minimized.
Figure~\ref{fig:karate} shows the results obtained with the Karate club dataset with our exact algorithm. We should stress that the labeling of the discovered \emph{subgraph} might be different from the labeling that maximizes strong edges for the whole graph. In Figure~\ref{fig:karate}~(c), we see that $4$ weak edges have been turned into strong while the labeling of the remaining edges is unchanged.

We show that our problem is \np-hard when $\lambda < 1$, and even inapproximable when $\lambda = 0$ since our problem then reduces to the \prbmaxclique problem. However, the $\lambda =  1$ case is solvable in polynomial time. To solve the problem,
we first propose an exact algorithm based on integer linear programming which runs in exponential time.  We consider four other heuristics that run in polynomial time in the size of the input graph:
We propose a linear linear programming based heuristic in Section~\ref{sec:lp}
and a greedy algorithm in Section~\ref{sec:greedy}.
We also propose two straightforward algorithms that combine the existing algorithms for solving the densest subgraph problem and finding STC-compliant labeling in an entire graph.

The remainder of the paper is organized as follows.
First, we introduce preliminary notation and our problem in Section~\ref{sec:prel}.
Next, we present the related work in Section~\ref{sec:related}.
Next, we show that our problem is \np-hard in Section~\ref{sec:compl} and then explain our algorithms in Section~\ref{sec:algorithm}.
Finally, we present our experimental evaluation in Section~\ref{sec:exp} and conclude the paper with a discussion in Section~\ref{sec:conclusions}.
\section{Preliminary notation and problem definition}\label{sec:prel}

We begin by providing preliminary notation and formally defining our problem.

Our input is an unweighted graph  $G = (V, E)$, and we denote the number of nodes and edges by $n = \abs{V}$ and  $m = \abs{E}$. 
Given a graph $G = (V, E)$ and a set of nodes $U \subseteq V$, we define $E(U) \subseteq E$ to be the subset of edges having both endpoints in $U$. We denote the degree of vertex $v$ as $\deg(v)$. We denote the set of adjacent edges of vertex $v$ as $N(v)$.  

We want to label the set of edges $E$ as either {\em strong} or {\em weak}. 
To perform the labeling,  we use the \emph{strong triadic closure} ({\stc}) property~\cite{sintos2014using}.  We say that the graph is \stc satisfied if for any given triplet of vertices, $(x, y, z)$ the following holds: if  $(x, y)$ and $(y, z)$ are connected and are labeled as strongly connected, then the edge $(x, z)$  always exists at least as a  weak edge. 

We call a triplet of vertices $(x, y, z)$  a {\em wedge}, if $(x, y) \in E$, $(y, z) \in E$, and $(x, z) \notin E$.
A \emph{wedge graph} $Z(G)$ consists of the edges of the graph $G$ that contribute to at least one wedge as its vertices. If the two edges $e_1$ and $e_2$ of $G$ form a wedge, we add an edge between the two nodes in $Z(G)$ that corresponds to $e_1$ and $e_2$, that is, each edge of $Z(G)$ corresponds to a wedge of $G$.

Given a graph $G = (V, E)$ and a labeling $L$ of the edges as {strong} or  {weak}, we write $E_s(U, L)$ and  $E_w(U, L)$ to be the set of {strong} and {weak} edges of the graph induced by a set of vertices $U$.
We also write $m_s(U, L) = \abs{E_s(U, L)}$ and $m_w(U, L) = \abs{E_w(U, L)}$.
Finally, for a vertex $u \in U$, we define a strong and weak degree, $\deg_s(u, U, L)$ and $\deg_w(u, U, L)$ to be the number of strong or weak edges in $U$ adjacent to $u$. We may omit $L$ or $U$ from the notation when it is clear from the context.

Given a weight parameter $\lambda$ where $\lambda \in [0,\;1]$ and a label assignment $L$, we define a score
\[
	\score{U, L; \lambda} =\frac{m_s(U, L) + \lambda m_w(U, L)} {\abs{U}} \quad .
\]
We may omit $L$ or $\lambda$ when it is clear from the context.

We consider the following optimization problem. 

\begin{problem}[\prbstrwk]
\label{pr:label-subgrap-str-wk}
Given a graph $G = (V,E)$ and a weight parameter $\lambda$, find a subset of vertices $U \subseteq V$ and a labeling $L$ of the edges as {strong} or {weak}  such that the \stc property is satisfied in $(U, E(U))$ and  $\score{U, L; \lambda}$ is maximized.
\end{problem}

Note that when $\lambda = 1$ then the labeling does not matter, and \prbstrwk reduces to dense subgraph discovery, that is, finding $U$ with the largest ratio $\abs{E(U)}/\abs{U}$. On the other hand, if $\lambda = 0$, then we only take into account strong edges; we will show in Section~\ref{sec:compl} that in this case, \prbstrwk is equal to finding a maximum clique.

Note also that the labeling depends on the underlying subgraph $U$, that is, we need to find $U$ and $L$ simultaneously.
\section{Related work}\label{sec:related}

\textbf{The strong triadic closure (\stc) property:}
As an early work in this line of research, \citet{sintos2014using} considered the problem of minimizing the number of weak edges (analogously maximizing the number of strong edges) while labeling the edges compliant with the \stc property.
We will refer to this optimization problem as \prbminSTC.

\citet{sintos2014using} showed
that \prbminSTC is equivalent to solving a minimum vertex cover, which we denote by \prbcovermin, in a wedge graph $Z(G)$.
In \prbcovermin, we search for a minimum number of nodes such that at least one endpoint of each edge is in the set.

\citet{sintos2014using} proposed the following algorithm for \prbminSTC, which we denote by \algminstc.
Given a graph $G$ they first construct the wedge graph $Z(G)$ and find its vertex cover. Next, they label the set of edges of $G$ that corresponds to the set of vertices in the vertex cover as weak, and the remaining edges as strong. Since \prbcovermin is \np-hard, they approximate it with a maximal matching algorithm~\citep{clarkson1983modification}. The algorithm picks an arbitrary edge and adds both endpoints of the edge to the cover, and the edges incident to both endpoints are deleted. It continues until no edges are left. This algorithm outputs a {\em maximal matching} which is known to yield a $2$-approximation since at least one endpoint should be in the cover.

The problem of finding a labeling of edges that maximizes the number of strong edges while satisfying the \stc property is \np-hard for general graphs \cite{sintos2014using} and split graphs \cite{konstantinidis2020maximizing}. 
Nevertheless, it becomes polynomial-time solvable for proper interval
graphs~\cite{konstantinidis2020maximizing}, cographs~\cite{konstantinidis2018strong}, and trivially perfect graphs~\cite{konstantinidis2020maximizing}.
Given communities, \citet{rozenshtein2017inferring} considered the problem of inferring strengths while minimizing \stc violations with additional connectivity constraints.
\citet{oettershagen2022inferring} extended the idea of inferring tie strength for temporal networks.
\citet{matakos2022strengthening} considered the problem of strengthening edges to maximize \stc violations, which they consider as opportunities to build new connections.
\citet{adriaens2020relaxing} formulated both minimization and maximization versions of \stc problems as linear programs.

\textbf{The dense subgraph discovery:}
Finding dense subgraphs is a core problem in social network analysis. 
Given a graph, the densest subgraph problem is defined as finding a subgraph with the highest average degree density~(twice the number of edges divided by the number of nodes).
Finding the densest subgraph for a single graph has been extensively studied~\cite{goldberg1984finding,charikar2000greedy,khuller2009finding}.
\citet{goldberg1984finding} proposed an exact,
polynomial time algorithm that solves a sequence of min-cut instances. We will refer to this algorithm as \algdg. \citet{asahiro2000greedily} provided a greedy algorithm and 
 \citet{charikar2000greedy} proved that their greedy algorithm gives a $1/2$-approximation, showed how to implement the algorithm using Fibonacci heaps, and devised a linear-programming formulation of the problem. The idea of the greedy algorithm is that at each iteration, a vertex with the minimum degree is removed, and then the densest
subgraph among all the produced subgraphs is returned as the solution. We will refer to this algorithm as \algdc. It has also been extended for multiple graph snapshots~\cite{jethava2015finding,semertzidis2019finding,arachchi2023jaccard}. The problem has been studied in a streaming setting~\cite{bhattacharya2015space}. 
To the best of our knowledge, this is the first attempt to consider the notion of density together with the \stc property.

In addition to degree density (a.k.a average degree), alternative types of density measures have also been considered previously such as triangle density and $k$-clique density~\cite{tsourakakis2015k}.
Triangle density is defined as the ratio between the number of triangles and the number of vertices of the subgraph.
The definition that will be used in this paper is the ratio between the number of induced edges and nodes.
Adopting our problem to other density measures is left open as future work.

\section{Computational complexity}\label{sec:compl}
In this section, we analyze the computational complexity of the \prbstrwk problem by showing its \np-hardness when $\lambda < 1$ and its connection to the \prbmaxclique problem, when $\lambda = 0$. 
At the other extreme, when $\lambda = 1$, \prbstrwk is equivalent to the problem of finding the densest subgraph, which can be solved in polynomial time using the algorithm presented by~\citet{goldberg1984finding}.
 
\begin{proposition}
\label{prop:np_hardness}
For $\lambda=0$, \prbstrwk is \np-hard.
\end{proposition}

The proof (given in Appendix~\ref{appendix:complexity}) also shows that the maximum clique 
yields precisely the optimal score 
while satisfying the \stc property. As a consequence, combining this with the inapproximability results for \prbmaxclique~\cite{zuckerman2006linear} gives the next result (see Appendix~\ref{appendix:complexity} for the proof).

\begin{proposition}
\label{prop:inapproximability}
For $\lambda=0$, \prbstrwk does not have any polynomial time approximation algorithm with an approximation ratio better than $n^{1-\epsilon}$ for any constant $\epsilon > 0$, unless $\p = \np$.
\end{proposition}

Note that when $\lambda > 0$, we can obtain a $\frac{1}{\lambda}$-approximation by finding the densest subgraph with nodes $U$ and using a labeling $L$ that labels each edge as weak. Compared to an optimal solution $U^*$ with labeling $L^*$, we then obtain the score 
\begin{equation*}
\begin{split}
q(U, L; \lambda) & = \frac{\lambda\abs{E(U)}}{\abs{U}} \geq \frac{\lambda\abs{E(U^*)}}{\abs{U^*}} = \lambda\frac{m_s(U^*, L^*)+m_w(U^*, L^*)}{\abs{U^*}} \\
& \geq \lambda\frac{m_s(U^*, L^*)+\lambda m_w(U^*, L^*)}{\abs{U^*}} = \lambda q(U^*, L^*; \lambda).
\end{split}
\end{equation*}

In summary, \prbstrwk is inapproximable when $\lambda = 0$ but solvable in polynomial time when $\lambda = 1$. Finally, we state that \prbstrwk is also \np-hard for $0 < \lambda < 1$.

\begin{proposition}
\label{prop:nphard2}
\prbstrwk is \np-hard for $0 < \lambda < 1$.
\end{proposition}

The proof of Proposition~\ref{prop:nphard2} is in Appendix~\ref{appendix:complexity}.

\section{Algorithms}\label{sec:algorithm}

In this section, we present five algorithms to find a good subgraph for our \prbstrwk problem. 
First, we propose an algorithm based on integer linear programming that finds a near-optimal or exact solution for our problem in Section~\ref{sec:ip}. Next, we state a polynomial time algorithm that solves a linear program in Section~\ref{sec:lp} followed by three heuristics presented in Sections~\ref{sec:stc-dense} and~\ref{sec:greedy}.

\subsection{Exact solution using integer programming}\label{sec:ip}

In this section, we present an integer linear programming~(ILP) based algorithm that can be used to find an exact solution for \prbstrwk.
To solve \prbstrwk we need the following auxiliary problem. The proofs for this section are given in Appendix~\ref{appendix:ip}.

\begin{problem}[$\prbstrwk(\alpha)$]
\label{pr:label-subgrap-str-wk-alpha}
Given a graph $G = (V, E)$, a weight parameter $\lambda$, and a number $\alpha$, find a subset of vertices $U \subseteq V$ and a labeling $L$ of the edges such that the \stc property is satisfied in $(U, E(U))$ and $m_s(U, L) + \lambda m_w(U, L) - \alpha \abs{U}$ is maximized.
\end{problem}

The following proposition, which is an instance of fractional programming~\citep{dinkelbach1967nonlinear}, shows the relationship between $\prbstrwk(\alpha)$ and \prbstrwk.

\begin{proposition}
\label{prop:frac}
Let $U(\alpha)$ and $L(\alpha)$ be the subgraph and the corresponding labeling solving $\prbstrwk(\alpha)$. Similarly, let $U^*$ with labeling $L^*$ be the solution for \prbstrwk. Write $\alpha^* = \score{U^*, L^*}$. If $\alpha > \alpha^*$, then $U(\alpha) = \emptyset$. If $\alpha < \alpha^*$, then $U(\alpha) \neq \emptyset$ and $\score{U(\alpha), L(\alpha)} > \alpha$.
\end{proposition}

We can use the proposition to solve \prbstrwk: we find the (almost) largest $\alpha$ for which $\prbstrwk(\alpha)$ yields a nonempty solution. Then $\prbstrwk(\alpha)$ for such $\alpha$ yields an (almost) optimal solution.

We can solve $\prbstrwk(\alpha)$ with an integer linear program,
\begin{alignat}{3}
\textsc{maximize}&\quad&\sum x_{ij} + \lambda \sum z_{ij}  & - \alpha \sum y_i \label{obj} \textstyle  & 
	\\  
	\textsc{subject to}&& x_{ij} + z_{ij} & \leq y_i &  ij & \in E \label{ip_con_1}\\
	&&x_{ij} + z_{ij} & \leq y_j &  ij & \in E \label{ip_con_2}\\
        &&x_{ij}  + x_{jk} & \leq y_j &  (i, j, k) & \in Z \label{ip_con_3}\\
        &&x_{ij}, z_{ij} & \in \{0,1\} &  ij & \in E \label{var_x_z}\\
        &&y_{i} & \in \{0,1\} &  i & \in V \label{var_v} \quad.
\end{alignat}
Here, $G = (V, E)$ is the input graph and $Z$ is the set of all wedges in $G$.

To see why this program solves $\prbstrwk(\alpha)$,
let $S \subseteq V$ and $L$ be a solution to our $\prbstrwk(\alpha)$. The indicator variable $y_i$ denotes whether the node $i \in S$ or not. The indicator variables $x_{ij}$ and $z_{ij}$ denote if the edge $(i,j)$ is strong or not and $(i,j) $ is weak or not, respectively.
Constraints~\ref{ip_con_1}-\ref{ip_con_2} guarantee that each edge within $S$ is labeled either as strong or weak. Constraint~\ref{ip_con_3} ensures that the \stc constraint is satisfied.  

Proposition~\ref{prop:frac}
allows us to maximize $\alpha$
with a binary search. Here, we set the initial interval $(L, U)$ to $L = 0$ and $U = \frac{n - 1}{2}$, and keep halving the interval until $\abs{ U - L} \leq  \epsilon L$, where $\epsilon > 0$  is an input parameter, and return the solution to $\prbstrwk(L)$. We refer to this algorithm as \algip.
Next, we state that \algip yields an approximation guarantee of  $1/(1+\epsilon)$.

\begin{proposition}
Assume a graph $G = (V, E)$, $\lambda \in [0,\;1]$, and $\epsilon > 0$. 
Let $\alpha$ be the score of the solution returned by $\algip$ and let $\alpha^*$ be the optimal score of \prbstrwk. Then $\alpha \geq \alpha^*/(1 + \epsilon)$.
\label{prop:opt-ip-approx}
\end{proposition}

Next, we will show that if $\epsilon$ is small enough, we are guaranteed to find the exact solution.

\begin{proposition}
Assume a graph $G$ with $n$ nodes. Assume that the weight parameter $\lambda$ is a rational number $\lambda = \frac{a}{b}$. Then, if we set   $\epsilon = \frac{2}{bn^3}$, \algip returns an exact solution for the \prbstrwk problem in $\bigO{\log n + \log b}$ number of iterations.
\label{prop:opt-ip-exact}
\end{proposition}

\algip requires $\bigO{\log n - \log \epsilon}$ iterations, solving an integer linear program in each round. Note that solving an
ILP is \np-hard~\cite{Schrijver1998}, and the fastest known algorithm to solve an ILP exactly runs in ${\log h}^{\bigO{h}}$ time where $h$ is the number of variables~\cite{reis2023subspace}. In practice, we can solve $\prbstrwk(\alpha)$ for moderately sized graphs but for larger graphs solving the ILP becomes computationally infeasible.

This approach is related to two prior works. First, the algorithm by \citet{goldberg1984finding} for finding the densest subgraph problem uses a similar approach, except without the variables $z_{ij}$. In such a case, the program can be solved exactly with a minimum cut.
Secondly, \citet{adriaens2020relaxing} use a linear program with similar wedge constraints to approximate \prbminSTC.

\subsection{Algorithm based on linear programming}\label{sec:lp}

In this section, we present an algorithm, named \alglpstc, based on a linear program obtained by relaxing the integrality requirements of the integer linear program given in the previous section. More specifically,
we first find a fractional solution by solving a linear program~(LP) and then derive a good subgraph via a rounding algorithm. Note that solving linear programs can be done in polynomial time~\citep{karmarkar1984new, van2020deterministic}, and solvers are efficient in practice. The proofs for this section are given in Appendix~\ref{appendix:lp}.

Consider a relaxed version of the ILP, where we replace the constraints in Eqs.~\ref{var_x_z}--\ref{var_v} with $x_{ij}, z_{ij} \in [0, 1]$ and $y_i \in [0, 1]$. We will refer to this optimization problem as $\prbstrwkr(\alpha)$.

Note that we used $\prbstrwk(\alpha)$ combined with the binary search to solve \prbstrwk. We can define a relaxed version of \prbstrwk which then can be analogously solved with $\prbstrwkr(\alpha)$.

\begin{problem}[\prbstrwkr]
Given a graph $G = (V,E)$, a weight parameter $\lambda$, find a nonnegative set of variables $x_{e}$, $y_i$, $z_{e}$, where $e \in E$ and $i \in V$ maximizing
\[
    r(x, y, z) = \frac{\sum x_e + \lambda \sum z_e}{ \sum y_i}
    \quad\text{such that}\quad\text{Eqs.~\ref{ip_con_1}--\ref{ip_con_3} hold}.
\]
\end{problem}

\prbstrwkr is a relaxed version of \prbstrwk:
if we were to require that the variables in \prbstrwkr are binary numbers, then the problems become equivalent. The next proposition is an analog to Proposition~\ref{prop:frac}.

\begin{proposition}
\label{prop:frac2}
 Let $(x^*, y^*, z^*)$ be a solution to \prbstrwkr. Write $\alpha^* = r(x^*, y^*, z^*)$.
Similarly, let $(x(\alpha), y(\alpha), z(\alpha))$ be a solution $\prbstrwkr(\alpha)$.
 If $\alpha > \alpha^*$, then $\sum y_i(\alpha) = 0$. On the other hand, if $\alpha < \alpha^*$, then $\sum y_i(\alpha) > 0$ and $r(x(\alpha), y(\alpha), z(\alpha)) > \alpha$.
\end{proposition}

Proposition~\ref{prop:frac2} allows us to solve \prbstrwkr with $\prbstrwkr(\alpha)$ and a binary search, similar to \algip. However, we can solve \prbstrwkr directly with a single linear program, that is,

\begin{alignat*}{3}
	\textsc{maximize}&\quad& \sum x_{ij} & + \lambda  \sum z_{ij}   & 
	\\  
	\textsc{subject to} &&  x_{ij} + z_{ij} & \leq y_i &  ij & \in E \\
	&& x_{ij} + z_{ij} & \leq y_j &  ij & \in E\\
   && \sum y_i & = 1 \\
    && x_{ij}  + x_{jk} & \leq y_j & (i, j, k) & \in Z\\
        && x_{ij}, z_{ij} & \geq 0 &  ij & \in E\\
        && y_{i} & \geq 0 &  i & \in V
\end{alignat*}

\begin{proposition}
\label{prop:lp}
The LP given above solves \prbstrwkr.
\end{proposition}

Our LP is related to the LP proposed  by~\citet{charikar2000greedy}, which is used to solve the densest subgraph problem exactly. 
We extend Charikar's LP by adding strong edges and additional \stc constraints.
Another related work is the LP proposed by~\citet{adriaens2020relaxing} which provides a $2$-approximation for \prbminSTC using similar wedge constraints.

\textbf{Rounding phase}:
Next, we describe the heuristic used to obtain the subgraph and the labeling from the variables. 
Let ($x^*$, $y^*$, $z^*$) be the solution to \prbstrwkr.
First we define a collection of sets $\calS =\{S_1, S_2, \ldots , S_n\}$ where $S_j = \{ i: y^{*}_i \geq  y_j^{*}\}$.
Then we enumerate over the collection of subgraphs induced by $\calS$. 

For each $S_j$, we initially set all the edges as weak. Then we enumerate over each edge $e \in E(S_j)$ starting from the largest $z_{e}^*$. Each edge $e$ is labeled as strong if the \stc property is not violated. This means that we check if there is any edge adjacent to any of the endpoints of $e$ which is already labeled as strong and still creates a wedge with $e$. We continue the same process for all the edges $e \in E(S_j)$ in the descending order of its $z_{e}^*$ value. Finally, out of all the subgraphs we pick the subgraph and the labeling that maximizes our score.

Constructing a labeling for a single $S_j$ amounts to enumerating over the wedges in $\bigO{nm}$ time, leading to a total time of $\bigO{n^2m}$ for the rounding.

\subsection{Label, find the densest subgraph, and relabel} \label{sec:stc-dense}
Next, we explain two algorithms that combine the existing methods for finding the densest subgraph and finding the STC-compliant labeling in an entire graph.

The approach is as follows.
First, we label the edges of the entire graph using \algminstc (see Section~\ref{sec:related}). Then we construct a weighted version of the graph assigning a weight of $1$ for strong edges and a weight of $\lambda$ for weak edges. Next, we search for the densest subgraph using \algdg or \algdc (see Section~\ref{sec:related}) in the new weighted graph.  
Finally, we relabel {\em only} the subgraph induced by the returned solution. Relabeling is used to improve the score since some of the edges might be marked as weak since they contributed to certain wedges in the original graph, nevertheless, some edges no longer contribute to all of those wedges.
The pseudo-code for this method is given in Algorithm~\ref{alg:dense}. We call the algorithm as \algdenseg or \algdensec
based on the subroutine used in Line~\ref{alg:dense-subgraph} of Algorithm~\ref{alg:dense}.

\begin{algorithm}[t!]
\caption{$\algdenseg(G, \lambda)$ and $\algdensec(G, \lambda)$, both find a subgraph $U$ and a labeling $L$ with good $\score{U, L; \lambda}$.}
\label{alg:dense}
    $L \define \algminstc(G)$\;
    $H \define$ the weighted graph by setting $1$ to strong edges and $\lambda$ to  weak edges of $G$\;
    $U \define  \algdg(H)$~\cite{goldberg1984finding}  or $\algdc(H)$~\cite{ charikar2000greedy}\label{alg:dense-subgraph}\;
    $L' \define \algminstc(G(U))$\;
    \Return subgraph $U$\ and its labeling $L'$\;
\end{algorithm}

Next, we present the computational complexities of the \algdensec and \algdenseg algorithms.

\begin{proposition}
Assume a graph $G$ with $n$ nodes and $m$ edges. Assume that the wedge graph $Z(G)$ contains $n'$ nodes and $m'$ edges.
Then the running time of \algdensec is in
\begin{align*}
	\bigO{\sum_{v \in V} \deg(v)^2 + (m + n \log n) + (m' + n')}  \subseteq \bigO{nm} \quad.
\end{align*}
\label{prop:dense}
\end{proposition}

\begin{proof}
The number of wedges in $G$, and hence the number of edges in $Z(G)$ is in $\bigO{\sum_{v \in V} \deg(v)^2} \subseteq \bigO{n\sum_{v \in V} \deg(v)}  \subseteq \bigO{nm}$. 
The number of vertices, $n'$, in the wedge graph $Z(G)$ is in $\bigO{m}$.
\algminstc estimates the minimum vertex cover problem with a maximum matching for $Z(G)$ and the subgraph, which has a running time of $\bigO{n' + m'}$ when the adjacency list representation is used for the graph~\cite{cormen2022introduction}. We can execute
\algdc in $\bigO{m + n \log n}$ time. The claim follows.
\end{proof}

\begin{proposition}
Assume a graph $G$ with $n$ nodes and  $m$ edges. Assume that the wedge graph $Z(G)$ contains $n'$ nodes and $m'$ edges.
Then the running time of \algdenseg  is in
\begin{align*}
\bigO{mn + n (n + m) \log n + (m' + n')} \subseteq \bigO{ mn \log n} \quad.
\end{align*} 

\label{prop:dense-goldberg}
\end{proposition}

\begin{proof}
The only change compared to Proposition~\ref{prop:dense} is that we are using an exact algorithm instead of an approximation algorithm for finding the densest subgraph.
The exact algorithm for an edge-weighted graph takes $\bigO{ M(n, n + m)\log n}$ time, and $M(n, n + m)$ is the time taken to solve the min-cut problem for a graph with $n$ number of nodes and $(n + m)$ number of edges.
It takes $\bigO{n (n + m)}$ to find the minimum cut
\cite{orlin2013max}. The claim follows.
\end{proof}

\subsection{Peeling with continuous relabeling}\label{sec:greedy}

The \algdensec algorithm, given in the previous section, first finds a labeling and then uses \algdc that constructs a set of subgraphs among which the subgraph with the highest score is selected. During this search, the labeling remains fixed. Our final algorithm modifies this approach by relabeling the graph as we are constructing the subgraphs.

Our approach is as follows. 
We start from the whole graph $G$ and label the edges as either strong or weak using \algminstc. 
Given a labeling $L$ and a subgraph $U$, let the weighted degree for a vertex $\scorev{v, U, L, \lambda}$ be defined as the sum of strong edges and weak edges in $U$ incident to $v$ weighted by $\lambda$, i.e., $\scorev{v, U, L} = \deg_s(v, U, L) + \lambda \deg_w(v, U, L)$ We drop $L$, $U$ or $\lambda$ when it is clear from the context.
At each iteration, we delete the node that has the minimum weighted degree $\scorev{v}$. After removing each vertex we relabel the remaining set of edges. Finally, we choose the subgraph $U$ which corresponds to the maximum score \score{U, \lambda} out of all the iterations.
The naive version for this method is given in Algorithm~\ref{alg:greedy}.

\begin{algorithm}[t!]
\caption{$\alggreedy(G, \lambda)$, finds a subgraph $U$ and a labeling $L$ with good $\score{U, L; \lambda}$}
\label{alg:greedy}
    $U \define V$\;
    \While {there are nodes} {
        $L \define \algminstc(G(U))$\label{alg:minstc-step}\;
	$u \define \displaystyle\arg\min_{ v \in U} \scorev{v, U, L, \lambda}$\label{alg:score-step}\;
	$U \define U \setminus \set{u}$\;
    }
	\Return best tested $U$ and its labeling $L$\;
\end{algorithm}

Next, we explain several tricks to speed up the naive implementation of Algorithm \ref{alg:greedy}. 
We focus on updating the wedge graph, modifying the minimum vertex cover, and updating individual scores of each vertex without computing them from scratch.

\textbf{Maintain wedge graph}: Note that on Line~\ref{alg:minstc-step} of Algorithm~\ref{alg:greedy}, we need to repeatedly construct a wedge graph to solve \prbminSTC.
We can avoid this by maintaining the existing wedge graph as vertices are deleted.

When a node is deleted we need to consider only {\em deleting} respective edges in the wedge graph since new wedges cannot be introduced. Note that an edge in the original graph~$G$ corresponds to a node in the wedge graph $Z(G)$ and edges in $Z(G)$ represent wedges in $G$. Next, we state how to maintain $Z(G)$ when a vertex is deleted in Proposition \ref{prop:wedge}.

\begin{proposition}
\label{prop:wedge}
Let $G = (V, E)$ be a graph.
Let $v$ be a vertex in $G$. Define $G' = G(V \setminus \set{v}, E \setminus N(v))$, where $N(v)$ is the set of adjacent edges of vertex $v$ in $G$.
Then, a new wedge graph $Z(G')$ is formed by deleting the vertices in $Z(G)$ corresponding to $N(v)$. 
\end{proposition}
We omit the straightforward proof.

\textbf{Dynamic vertex cover using maximal matching}:
Next, we consider updating the vertex cover after a vertex deletion. 

Recall that we use maximum matching to approximate the vertex cover in \algminstc.
Given a maximal matching of the current graph, \citet{ivkovic1993fully} presented a simple algorithm to maintain the cover when an {\em edge} is deleted or inserted. Here we modify their algorithm slightly to adapt to a node deletion from $G$. 
Let us consider the case where the vertex $v$ is deleted from the original graph $G$. 
Note that $N(v)$ is a set of edges in $G$ which corresponds to a subset of nodes in $Z(G)$.
According to Proposition~\ref{prop:wedge},
the set of nodes corresponding to $N(v)$ should be deleted from the wedge graph $Z(G)$ to compensate for the deleted vertex. We assume that a maximal matching~$M$ of $Z(G)$ is given.

The algorithm is as follows.
We iterate over the elements in $N(v)$ and pick a node $a \in N(v)$ in $Z(G)$. We then test whether there is an edge $(a, b)$ in $M$ for some $b$. There can be only one, and if there is, we delete it. Upon such deletion, $M$ may no longer be maximal since $b$ may have a single adjacent edge that can be added. We search for such an edge and add it if one is found. 

The pseudocode is given in Algorithm~\ref{alg:dynamic-vc-node-pov}.
Algorithm~\ref{alg:dynamic-vc-node-pov} still produces a maximal matching; thus a $2$-approximation for \prbcovermin is guaranteed.

\begin{algorithm}[t!]
\caption{$\algdynvc(M,  v)$, maintains a vertex cover (a maximal matching $M$) when a node $v$ is deleted}
\label{alg:dynamic-vc-node-pov}
    \ForEach {$a \in N(v)$} {
        \If {there is $b$ such that $(a, b) \in M$} {
            delete $(a, b)$ from $M$\;
            \If {$b \notin N(v)$ \AND there is $t \notin N(v)$ such that $t$ is not an endpoint of any edge in $M$ and $(b, t) \in E(Z(G))$} {
                add $(b, t)$ to $M$\;
            }
        }
    }
    \Return $M$\;
\end{algorithm}

\textbf{Speeding the vertex selection}:
We can speed up finding the next vertex by maintaining $\scorev{v, \lambda}$ in a priority queue.
Once a vertex is deleted, we need to update the degree of its neighboring nodes. Also, we may need to update the weighted degree of the affected vertices if the vertex cover of $Z(G)$ changes. However, the number of changed edges in the vertex cover is constant.
The final version of the algorithm is presented in Algorithm~\ref{alg:greedy-fast-2}. 

\begin{algorithm}[ht!]
\caption{$\alggreedyfastest(G, \lambda)$, finds a subgraph $U$ and a labeling $L$ with good $\score{U, L; \lambda}$}
\label{alg:greedy-fast-2}
    $L \define \algminstc(G)$\;
    $P \define$ priority queue where each node is ranked by $\scorev{v, L}$\;
    $U \define V$\;
    \While {there are nodes} {
	$u \define \displaystyle\arg\min_{ v \in U} \scorev{v, U, L} \label{alg:score-step-pq}$\;
        $U \define U \setminus \set{u}$\;
        Update the wedge graph $Z((U, E(U)))$\;
        Update labeling $L$ using Algorithm~\ref{alg:dynamic-vc-node-pov}\;
        Update $P$\;
    }
    \Return best tested $U$ and its labeling $L$\;
\end{algorithm}

Next, we state the computational complexity of \alggreedyfastest.

\begin{proposition}
Assume a graph $G$ with $n$ nodes and  $m$ edges. Assume that the wedge graph $Z(G)$ contains $n'$ nodes and $m'$ edges.
Then the running time of \alggreedyfastest is in
\[
\bigO{ mn  + (n' + m') + m \log n + nm} \subseteq \bigO{nm} \quad.
\]
\label{prop:greedy}
\end{proposition}
\begin{proof}
Let $G_i$ be the graph at $i$th iteration. Consider deleting vertex $u$ from $G_i$.
Upon deletion, we need to update the priorities of the affected nodes in the queue.

When $u$ is deleted from $G_i$, we need to delete the set of nodes in $Z(G)$ which corresponds to the adjacent edges of $u$. 
For each deleted vertex in $Z(G)$, there can be at most one adjacent edge that belongs to the existing matching set.
Therefore, to compensate for the edge that is removed from the maximal matching set, we need to add at most one edge to the matching.
The two endpoints of the newly added edge correspond to two edges in $G_{i+1}$.
Therefore, the total number of vertices that require updating priorities is at most $4$. 
Moreover, deleting one edge from the existing matching set will affect the priorities of at most $2$ vertices.

In summary, $\bigO{\deg u}$ nodes need to be updated when we delete $u$. Consequently, the total update time of $P$ is in $\bigO{m \log n}$.
Moreover, the total update time for $Z(G)$ is in $\bigO{n' + m'}$. Updating $M$ requires finding a new edge which may cost $\bigO{n}$ time, consequently, updating $M$ requires $\bigO{nm}$ total time. Finally, the update time for $(U, E(U))$ is in $\bigO{m}$.

Initially, constructing $Z(G)$ requires  $\bigO{\sum_{v \in V} \deg(v)^2} \subseteq \bigO{nm}$ time and \algvcm requires $\bigO{n' + m'} \subseteq \bigO{nm}$ time.

Combining these times proves the claim.
\end{proof}

\section{Experimental evaluation}\label{sec:exp}

\begin{table*}[t]
\caption{
Results of the experiments for synthetic and real datasets. 
Here, $\lambda$ is the weight parameter, columns $\acu$, $\ape$, $\agr$, $\alp$, and $\aip$ represent  \algdenseg, \algdensec, \alggreedyfastest, \alglpstc, and \algip, respectively, columns in $\score{}$ are the discovered scores, columns in $\abs{E_s} \%$ give the percentages of {\em strong} edges in discovered subgraph, and columns in time give the computational time.
}

\label{tab:real-result}
\setlength{\tabcolsep}{0pt}
\begin{filecontents*}{comb-results.csv}
name,lam,score-1,score-2,score-3,score-4,score-5,obj-LP,t-1,t-2,t-3,t-4,t-5,dis_r,LP-time,s1,s2,s3,s4,s5,w1,w2,w3,w4,w5,s_p-1,s_p-2,s_p-3,s_p-4,s_p-5,w_p-1,w_p-2,w_p-3,w_p-4,w_p-5,len-1,len-2,len-3,len-4,len-5,,,,,,,,,,,,
\dtname{Synthetic},0.8,21.271,21.279,21.279,21.682,,,6.9449,11.7206,6.6907,29.5796,,,,483,491,491,894,,4714,4706,4706,4303,,9.433237272,9.447758322,15.96428571,17.20223206,,90.56676273,90.55224168,84.03571429,82.79776794,,200,200,200,200,,,,,,,,,,,,,
,0.6,18.5,16.573,16.573,17.365,,,7.1063,11.2929,6.4006,36.8970,,,,703,491,491,887,,0,4706,4706,4310,,100,9.447758322,15.8591096,17.06753896,,0,90.55224168,84.1408904,82.93246104,,38,200,200,200,,,,,,,,,,,,,
,0.4,18.5,18.5,18.5,18.5,,,7.8458,9.0008,4.7263,16.9691,,,,703,703,703,703,,0,0,0,0,,100,100,100,100,,0,0,0,0,,38,38,38,38,,,,,,,,,,,,,
,0.2,18.5,18.5,18.5,18.5,18.5,,7.9825,7.0178,7.0178,12.9199,545.66022205,,,703,703,703,703,703,0,0,0,0,0,100,100,100,100,100,0,0,0,0,0,38,38,38,38,38,,,,,,,,,,,,
\dtname{Cora},0.8,2.298203593,2.579661017,2.331168831,2.653571429,2.689655172,2.83559322,61.3288,18.1445,1.1041,2.4078,43.25872493,,,75,21,79,39,52,386,164,350,137,130,5.15970516,32.51028807,10.02570694,22.15909091,28.57142857,94.84029484,67.48971193,89.97429306,77.84090909,71.42857143,167,59,154,56,58,0.03325878364,0.1081239041,0.03641456583,0.1142857143,13861,1711,11781,1540,461,185,429,176
,0.6,2.006666667,1.9875,1.838974359,2.142857143,2.294117647,2.549019608,61.2213,17.5068,0.9243,2.8034,661.06763411,,,59,9,109,33,39,2,38,416,120,65,81.81818182,74.14965986,7.349665924,21.56862745,37.5,18.18181818,25.85034014,92.65033408,78.43137255,62.5,30,16,195,49,34,0.1402298851,0.3916666667,0.02775574941,0.1301020408,435,120,18915,1176,61,47,525,153
,0.4,2,2,1.777777778,1.622222222,2.088888889,2.299115044,46.7298,14.6778,0.9442,2.9828,263.66145897,,,10,10,16,9,16,0,0,0,14,7,100,100,100,39.13043478,69.56521739,0,0,0,60.86956522,30.43478261,5,5,9,9,9,1,1,0.4444444444,0.6388888889,10,10,36,36,10,10,16,23
,0.2,2,2,1.777777778,1.32,2.02,2.143243243,42.3364,15.4789,0.8789,2.7869,257.23868322,,,10,10,16,6,20,0,0,0,3,1,100,100,100,66.66666667,95.23809524,0,0,0,33.33333333,4.761904762,5,5,9,5,10,1,1,0.4444444444,0.9,10,10,36,10,10,10,16,9
\dtname{Citeseer},0.8,3.866666667,4,3.99047619,4.118181818,4.2,4.418604651,85.5068,13.1678,0.5067,1.2697,18.31660223,,,4,8,7,21,30,82,100,96,87,78,8.888888889,6.542056075,17.94871795,19.44444444,27.77777778,91.11111111,93.45794393,82.05128205,80.55555556,72.22222222,18,22,21,22,22,0.5620915033,0.4675324675,0.4904761905,0.4675324675,153,231,210,231,86,108,103,108
,0.6,2.990697674,3.114285714,3.076190476,3.345454545,3.495238095,3.930232558,78.8847,15.4365,0.5239,1.2769,17.53544402,,,35,9,7,22,29,156,94,96,86,74,5.454545455,6.930693069,18.6440678,20.37037037,28.15533981,94.54545455,93.06930693,81.3559322,79.62962963,71.84466019,43,21,21,22,21,0.211517165,0.4904761905,0.4904761905,0.4675324675,903,210,210,231,191,103,103,108
,0.4,2.5,2.219047619,2.163636364,2.647619048,2.822222222,3.441860465,56.8552,14.5450,0.5637,1.2985,38.62018991,,,15,9,18,24,28,0,94,74,79,57,100,16.07142857,24.48979592,23.30097087,32.94117647,0,83.92857143,75.51020408,76.69902913,67.05882353,6,21,22,21,18,1,0.4904761905,0.3982683983,0.4904761905,15,210,231,210,15,103,92,103
,0.2,2.5,1.5,1.5,1.892307692,2.5,2.959090909,50.0570,15.0472,0.6801,1.2851,66.91253710,,,15,6,6,18,30,0,0,0,33,0,100,100,100,35.29411765,100,0,0,0,64.70588235,0,6,4,4,13,12,1,1,1,0.6538461538,15,6,6,78,15,6,6,51
\dtname{PGP},0.8,15.4,15.95555556,15.58604651,16.47555556,16.86222222,17.17142857,1909.0972,93.9528,5.9474,22.3922,415.94080114,3,431.7694559,122,158,123,275,362,656,700,684,583,496,15.68123393,18.41491841,15.24163569,32.05128205,42.19114219,84.31876607,81.58508159,84.75836431,67.94871795,57.80885781,42,45,43,45,45,0.9036004646,0.8666666667,0.8936877076,0.8666666667,861,990,903,990,778,858,807,858
,0.6,12,12.50864198,12.07162162,13.97333333,14.65777778,15.27692308,1912.1701,98.6242,6.6707,20.5164,520.25445175,9,563.0057421,300,268,631,285,362,0,1242,1926,573,496,100,17.74834437,24.67735628,33.21678322,42.19114219,0,82.25165563,75.32264372,66.78321678,57.80885781,25,81,148,45,45,1,0.4660493827,0.2350615922,0.8666666667,300,3240,10878,990,300,1510,2557,858
,0.4,11,10.5,10.5,11.44186047,12.55609756,13.38444444,1864.3785,98.7629,6.4192,20.3321,564.81102920,9,638.7495,253,231,231,282,360,0,0,0,525,387,100,100,100,34.94423792,48.19277108,0,0,0,65.05576208,51.80722892,23,22,22,43,41,1,1,1,0.8936877076,253,231,231,903,253,231,231,807
,0.2,11,10.5,10.5,10.46,12,12.11877934,1460.0874,100.2147,6.2800,27.8216,543.06912732,,,253,231,231,286,300,0,0,0,139,0,100,100,100,67.29411765,100,0,0,0,32.70588235,0,23,22,22,30,41,1,1,1,0.9770114943,253,231,231,435,253,231,231,425
\dtname{Email-EU},0.8,22.11928251,22.17077626,22.1719457,22.44464286,,24.81026786,149.5216,52.4970,22.2136,217.8661,,18,,75,129,132,438,,6072,5908,5960,5737,,1.220107369,2.136822925,2.1667761,7.093117409,,98.77989263,97.86317707,97.8332239,92.90688259,,223,219,221,224,,0.2483335353,0.2529010096,0.2505964624,0.2472373479,24753,23871,24310,24976,6147,6037,6092,6175
,0.6,16.67733333,16.77534247,16.7635514,17.29910714,,22.05357143,153.2032,54.2583,21.1262,266.6796,,,,78,129,123,425,,6124,5908,5774,5750,,1.25765882,2.136822925,2.085806342,6.882591093,,98.74234118,97.86317707,97.91419366,93.11740891,,225,219,214,224,,0.2461111111,0.2529010096,0.2587424861,0.2472373479,25200,23871,22791,24976,6202,6037,5897,6175
,0.4,11.24114833,11.40275229,11.37578475,12.19647577,,19.29911308,145.7109,55.0003,21.7625,316.7498,,,,83,141,140,445,,5666,5862,5992,5809,,1.443729344,2.348825587,2.283105023,7.115446114,,98.55627066,97.65117441,97.71689498,92.88455389,,209,218,223,227,,0.2644920869,0.2537944447,0.2477275482,0.2438111575,21736,23653,24753,25651,5749,6003,6132,6254
,0.2,8,6.017605634,5.973818182,7.16069869,,16.54746137,141.1825,57.6321,23.5424,297.9604,,,,136,233,212,474,,0,7380,7154,5829,,100,3.060554315,2.878088515,7.520228463,,0,96.93944569,97.12191149,92.47977154,,17,284,275,229,,1,0.189444085,0.1955142668,0.2414387497,136,40186,37675,26106,136,7613,7366,6303
\dtname{Facebook},0.8,62.02288557,62.26169154,63.33034826,61.48457711,,69.05820896,329.7502,82.6422,37.3082,447.1424,,39,298.0458908,641,881,955,1955,,14782,14542,14468,13468,,4.156130455,5.712247941,6.192050833,12.6758737,,95.84386955,94.28775206,93.80794917,87.3241263,,201,201,201,201,,0.7673134328,0.7673134328,0.7673134328,0.7673134328,20100,20100,20100,20100,15423,15423,15423,15423
,0.6,47.31442786,47.7920398,47.93930348,49.84,,61.38759305,339.5542,83.7318,36.0494,394.9769,,23,328.1041412,641,881,955,1916,,14782,14542,14468,13420,,4.156130455,5.712247941,6.192050833,12.49347939,,95.84386955,94.28775206,93.80794917,87.50652061,,201,201,201,200,,0.7673134328,0.7673134328,0.7673134328,0.7706532663,20100,20100,20100,19900,15423,15423,15423,15336
,0.4,33.5,33.54141414,33.38383838,36.30456853,,53.7191067,332.7295,81.1121,35.2141,570.4061,,103,344,2278,950,898,1914,,0,14228,14280,13095,,100,6.259059165,5.916458031,12.75234859,,0,93.74094084,94.08354197,87.24765141,,68,198,198,197,,1,0.7782392452,0.7782392452,0.7774267067,2278,19503,19503,19306,2278,15178,15178,15009
,0.2,33,19.2688172,19.36344086,23.178125,,46.16558603,322.8660,82.2929,34.2089,567.2840,,146,413,2211,954,976,1947,,0,13150,13128,12516,,100,6.764038571,6.920022689,13.46193736,,0,93.23596143,93.07997731,86.53806264,,67,186,186,192,,1,0.8197616972,0.8197616972,0.788776178,2211,17205,17205,18336,2211,14104,14104,14463
\dtname{LastFM},0.8,11.94098361,11.97460317,12,12.27868852,,13.3168,1494.4453,93.8324,11.9914,61.4496,,9,,34,44,52,137,,868,888,880,765,,3.76940133,4.721030043,5.579399142,15.18847007,,96.23059867,95.27896996,94.42060086,84.81152993,,61,63,63,61,,0.4928961749,0.4772145417,0.4772145417,0.4928961749,1830,1953,1953,1830,902,932,932,902
,0.6,8.976470588,9.155555556,9.190322581,9.672131148,,11.85079365,1485.2486,93.7915,12.0941,58.4179,,,,39,44,49,122,,698,888,868,780,,5.291723202,4.721030043,5.34351145,13.52549889,,94.7082768,95.27896996,94.65648855,86.47450111,,51,63,62,61,,0.5780392157,0.4772145417,0.4849286092,0.4928961749,1275,1953,1891,1830,737,932,917,902
,0.4,6.5,6.46557377,6.426229508,7.09122807,,10.47355372,1390.7393,92.2703,11.3798,63.3475,,14,,91,56,52,119,,0,846,850,713,,100,6.208425721,5.764966741,14.30288462,,0,93.79157428,94.23503326,85.69711538,,14,61,61,57,,1,0.4928961749,0.4928961749,0.5213032581,91,1830,1830,1596,91,902,902,832
,0.2,6.5,3.464516129,3.518644068,4.711111111,,9.347457627,1048.7187,92.6047,11.7972,61.8549,,,,91,40,42,124,,0,874,828,652,,100,4.376367615,4.827586207,15.97938144,,0,95.62363239,95.17241379,84.02061856,,14,62,59,54,,1,0.483342147,0.5084745763,0.5422781272,91,1891,1711,1431,91,914,870,776
\end{filecontents*}
\pgfplotstabletypeset[
    begin table={\begin{tabular*}{\textwidth}},
    end table={\end{tabular*}},
    col sep=comma,
	columns = {name,lam,score-2,score-3,score-1,score-4,score-5,s_p-2,s_p-3,s_p-1,s_p-4,s_p-5,t-2,t-3,t-1,t-4,t-5},
    columns/name/.style={string type, column type={@{\extracolsep{\fill}}l}, column name=\emph{Dataset}},
    columns/lam/.style={fixed, set thousands separator={\,}, column type=r, column name=$\lambda$},
    columns/score-1/.style={fixed, set thousands separator={\,}, column type=r, column name=$\ago$}, 
    columns/score-5/.style={fixed, set thousands separator={\,}, column type=r, column name=$\aip$}, 
    columns/len-1/.style={fixed, set thousands separator={\,}, column type=r, column name=$b$},
    columns/s_p-1/.style={fixed, set thousands separator={\,}, column type=r, column name=$\ago$}, 
    columns/w_p-1/.style={fixed, set thousands separator={\,}, column type=r, column name=$w$}, 
    columns/t-1/.style={fixed, set thousands separator={\,}, column type=r, column name=$\ago$,
        assign cell content/.code={%
            \ifstrequal{##1}{}{}{%
                \pgfkeyssetvalue{/pgfplots/table/@cell content}{%
                    \pgfmathtruncatemacro\minutes{##1/60}%
                    \ifnumcomp{\minutes}{=}{0}{
                        \pgfmathprintnumber{##1}s
                    }{%
                        \pgfmathtruncatemacro\seconds{##1 - \minutes*60}%
                        {\minutes}m{\seconds}s%
                    }%
                }%
            }%
        }%
    },
    columns/score-2/.style={fixed, set thousands separator={\,}, column type=r, column name=$\acu$}, 
    columns/len-2/.style={fixed, set thousands separator={\,}, column type=r, column name=$b$},
    columns/s_p-2/.style={fixed, set thousands separator={\,}, column type=r, column name=$\acu$}, 
    columns/w_p-2/.style={fixed, set thousands separator={\,}, column type=r, column name=$w$}, 
    columns/t-2/.style={fixed, set thousands separator={\,}, column type=r, column name=$\acu$,
        assign cell content/.code={%
            \ifstrequal{##1}{}{}{%
                \pgfkeyssetvalue{/pgfplots/table/@cell content}{%
                    \pgfmathtruncatemacro\minutes{##1/60}%
                    \ifnumcomp{\minutes}{=}{0}{
                        \pgfmathprintnumber{##1}s
                    }{%
                        \pgfmathtruncatemacro\seconds{##1 - \minutes*60}%
                        {\minutes}m{\seconds}s%
                    }%
                }%
            }%
        }%
    },
    columns/t-5/.style={fixed, set thousands separator={\,}, column type=r, column name=$\aip$,
        assign cell content/.code={%
            \ifstrequal{##1}{}{}{%
                \pgfkeyssetvalue{/pgfplots/table/@cell content}{%
                    \pgfmathtruncatemacro\minutes{##1/60}%
                    \ifnumcomp{\minutes}{=}{0}{
                        \pgfmathprintnumber{##1}s
                    }{%
                        \pgfmathtruncatemacro\seconds{##1 - \minutes*60}%
                        {\minutes}m{\seconds}s%
                    }%
                }%
            }%
        }%
    },
    columns/score-4/.style={fixed, set thousands separator={\,}, column type=r, column name=$\alp$}, 
    columns/len-4/.style={fixed, set thousands separator={\,}, column type=r, column name=$b$},
    columns/s_p-4/.style={fixed, set thousands separator={\,}, column type=r, column name=$\alp$}, 
    columns/s_p-5/.style={fixed, set thousands separator={\,}, column type=r, column name=$\aip$}, 
    columns/w_p-4/.style={fixed, set thousands separator={\,}, column type=r, column name=$w$}, 
    columns/t-4/.style={fixed, set thousands separator={\,}, column type=r, column name=$\alp$,
        assign cell content/.code={%
            \ifstrequal{##1}{}{}{%
                \pgfkeyssetvalue{/pgfplots/table/@cell content}{%
                    \pgfmathtruncatemacro\minutes{##1/60}%
                    \ifnumcomp{\minutes}{=}{0}{
                        \pgfmathprintnumber{##1}s
                    }{%
                        \pgfmathtruncatemacro\seconds{##1 - \minutes*60}%
                        {\minutes}m{\seconds}s%
                    }%
                }%
            }%
        }%
    },
    columns/score-3/.style={fixed, set thousands separator={\,}, column type=r, column name=$\ape$}, 
    columns/len-3/.style={fixed, set thousands separator={\,}, column type=r, column name=$b$},
    columns/s_p-3/.style={fixed, set thousands separator={\,}, column type=r, column name=$\ape$}, 
    columns/w_p-3/.style={fixed, set thousands separator={\,}, column type=r, column name=$w$}, 
    columns/t-3/.style={fixed, set thousands separator={\,}, column type=r, column name=$\ape$,
        assign cell content/.code={%
            \ifstrequal{##1}{}{}{%
                \pgfkeyssetvalue{/pgfplots/table/@cell content}{%
                    \pgfmathtruncatemacro\minutes{##1/60}%
                    \ifnumcomp{\minutes}{=}{0}{
                        \pgfmathprintnumber{##1}s
                    }{%
                        \pgfmathtruncatemacro\seconds{##1 - \minutes*60}%
                        {\minutes}m{\seconds}s%
                    }%
                }%
            }%
        }%
    },
    every head row/.style={
		before row={\toprule
		& &
		\multicolumn{5}{l}{$\score{}$} &
		\multicolumn{5}{l}{$\abs{E_s}\%$} &
		\multicolumn{5}{l}{time} \\
		\cmidrule(r){3-7}
		\cmidrule(r){8-12}
            \cmidrule(r){13-17}
		},
			after row=\midrule},
    every last row/.style={after row=\bottomrule},
]
{comb-results.csv}
\end{table*}

\begin{table}[t!]
\setlength{\tabcolsep}{0pt}
\caption{
Characteristics of real-world datasets. Here, $\abs{V}$ and $\abs{E}$ give the number of
vertices and edges, $\abs{V(Z)}$ and $\abs{E(Z)}$ are the number of vertices and edges in the wedge graph, $d = \abs{E(U)} / \abs{U}$
gives the density of the densest subgraph, and $d_{c}$ is the density induced by \prbmaxclique. 
}

\label{tab:real-data}
\begin{filecontents*}{real-world.csv}
name,n,m,nw,mw,dense,size,c,d
\dtname{Cora},2708,5278,5151,47411,3.142857143,56,5,2
\dtname{Citeseer},3264,4536,4192,23380,4.909090909,22,6,2.5
\dtname{PGP},10680,24316,23568,270433,19.06666666,45,25,12
\dtname{Email-EU},986,16064,16063,866833,27.56696401,224,18,8.5
\dtname{Facebook},747,30025,30022,1177951,76.73134264,201,68,33.5
\dtname{LastFM},7624,27806,27775,557781,14.79365079,63,15,7
\end{filecontents*}
\pgfplotstabletypeset[
    begin table={\begin{tabular*}{\columnwidth}},
    end table={\end{tabular*}},
    col sep=comma,
    columns = {name,n,m,nw,mw,dense,d},
    columns/name/.style={string type, column type={@{\extracolsep{\fill}}l}, column name=\emph{Dataset}},
    columns/n/.style={fixed, set thousands separator={\,}, column type=r, column name=$\abs{V}$},
    columns/m/.style={fixed, set thousands separator={\,}, column type=r, column name=$\abs{E}$},
    columns/nw/.style={fixed, set thousands separator={\,}, column type=r, column name=$\abs{V(Z)}$},
    columns/mw/.style={fixed, set thousands separator={\,}, column type=r, column name=$\abs{E(Z))}$},
    columns/dense/.style={fixed, set thousands separator={\,}, column type=r, column name=$d$}, 
    columns/d/.style={fixed, set thousands separator={\,}, column type=r, column name=$d_{c}$}, 
    every head row/.style={
		before row={\toprule},
			after row=\midrule},
    every last row/.style={after row=\bottomrule},
]
{real-world.csv}
\end{table}

Next, we evaluate our
algorithms experimentally. We first generate a synthetic dataset with a dense subgraph component and test how well our algorithms perform.
Next, we study the performance of the algorithms on real-world networks.
We implemented the algorithms in Python\footnote{The source code is available at \url{https://version.helsinki.fi/dacs/}.
\label{foot:code}} and performed the experiments using a 2.4GHz Intel Core i5 processor and 16GB RAM. In our experimental evaluation, we used Gurobi solver in Python to solve the ILPs and LPs associated with \algip and \alglpstc respectively.

\textbf{Synthetic dataset}:
We will now explain how the synthetic dataset was generated.
Given a vertex set $V$ of size 230, we split $V$ into dense and sparse components $D$ and $S$. Here, we randomly selected $D$ to have 38 nodes and $S$ to have 192 nodes.
We sampled the edges using a stochastic block model, with the edge probabilities
being $p_d = 1$, $p_s = 0.3$, and $p_c = 0.05$ for dense component, sparse component, and cross edges, respectively. 
The resulting graph had $5\,197$ edges, and the wedge graph had $5\,197$ nodes and $179\,100$ edges. The density of $D$ was $\abs{E(D)} / \abs{D} = 18.5$.

\textbf{Results using synthetic dataset}:
We report our results in Table~\ref{tab:real-result}. First, we see that all our algorithms find the ground truth by achieving a score of $18.5$ which is the density of our planted clique of size $38$ for example when $\lambda = 0.4$ and $\lambda = 0.2$. Note that \algip produced the results within an hour only for the $\lambda = 0.2$ case. Since \algip solves an ILP in each round, it was inefficient to run for the other $\lambda$ values and we stopped the execution after one hour.

As $\lambda$ increases, our algorithms tend to find a score greater than $18.5$ by deviating away from the planted clique. We also see that \alglpstc which solves a linear program runs significantly slower than \alggreedyfastest, \algdensec, and \algdenseg algorithms. 

Next, we study how the scores and the percentage of weak edges vary as a function of $\lambda$ as shown in Figure~\ref{fig:lam-q-syn}. We can observe that both \algdenseg and \algdensec produce equal scores whereas \algdensec and \alglpstc slightly underperform at $\lambda = 0.6$ and $\lambda = 0.5$ respectively as shown in Fig.~\ref{fig:sla}. Moreover, the \alglpstc slightly outperforms other algorithms when $\lambda \geq 0.7$. In terms of percentages of weak edges, all three algorithms produced the same decreasing trends according to Fig.~\ref{fig:slb}. 
There are no weak edges in the subgraphs produced by any of the algorithms when $\lambda \leq 0.4$ since scores are only contributed by the planted clique. Recall the connection to the maximum clique problem for $\lambda=0$ from Proposition~\ref{prop:np_hardness}.

Finally, we study the running time as a function of the number of edges $\abs{E}$ and the number of wedges $\abs{V(Z)}$ in Figure~\ref{fig:e-w-t}. We randomly generated $6$  datasets each with $5\,000$ nodes. The number of edges of the datasets uniformly ranges from $1 \times 10^4$ to $1.1 \times 10^5$.
We see that \algdenseg and \algdensec are the fastest while \alglpstc is the slowest.

\begin{figure}[t!]
\begin{subcaptiongroup}
\phantomcaption\label{fig:sla}
\phantomcaption\label{fig:slb}
\begin{center}
\begin{tabular}{ll}
\begin{tikzpicture}
\begin{axis}[xlabel={$\lambda$}, ylabel= {$\score{\cdot; \lambda}$},
    width = 4.5cm,
    height = 3.8cm,
    xmin = 0,
    xmax = 1,
    ymin = 15,
    ymax = 40,
    scaled y ticks = true,
    cycle list name=yaf,
    yticklabel style={/pgf/number format/fixed},
    no markers,
    legend entries = {\alggreedyfastest, \algdenseg, \algdensec,\alglpstc},
    legend pos = north west
]
\addplot table [x=x, y=y1, col sep=comma] {syn_lam.csv};
\addplot table [x=x, y=y2, col sep=comma] {syn_lam.csv};
\addplot table [x=x, y=y3, col sep=comma] {syn_lam.csv};
\addplot table [x=x, y=y4, col sep=comma] {syn_lam.csv};
\pgfplotsextra{\yafdrawaxis{0}{1}{15}{40}}
\end{axis}
\node[anchor=north east] at (0, -0.3) {(a)};
\end{tikzpicture}&
\begin{tikzpicture}
\begin{axis}[xlabel={$\lambda$}, ylabel= {$\abs{E_s}\%$},
    width = 4.5cm,
    height = 3.8cm,
    xmin = 0,
    xmax = 1,
    ymin = 0,
    ymax = 100,
    scaled y ticks = true,
    cycle list name=yaf,
    yticklabel style={/pgf/number format/fixed},
    no markers
]
\addplot table [x=x, y=p1, col sep=comma] {syn_lam.csv};
\addplot table [x=x, y=p2, col sep=comma] {syn_lam.csv};
\addplot table [x=x, y=p3, col sep=comma] {syn_lam.csv};
\addplot table [x=x, y=p4, col sep=comma] {syn_lam.csv};
\pgfplotsextra{\yafdrawaxis{0}{1}{0}{100}}
\end{axis}
\node[anchor=north east] at (0, -0.3) {(b)};
\end{tikzpicture}

\end{tabular}
\end{center}
\end{subcaptiongroup}
\caption{Scores and percentages of strong edges as a function of $\lambda$ for \dtname{Synthetic} dataset.
}
\label{fig:lam-q-syn}
\end{figure}
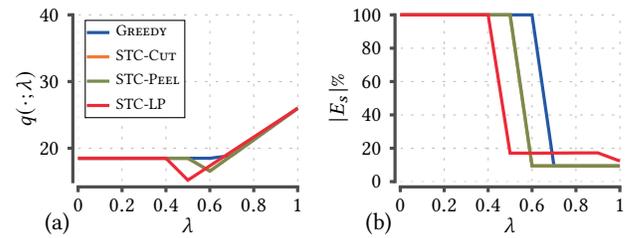

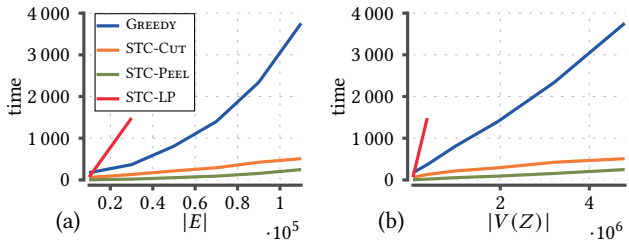
\begin{figure}[t!]
\begin{subcaptiongroup}
\phantomcaption\label{fig:eta}
\phantomcaption\label{fig:wtb}
\begin{center}
\begin{tabular}{ll}
\begin{tikzpicture}
\begin{axis}[xlabel={$\abs{E}$}, ylabel= {time},
    width = 4.4cm,
    height = 3.8cm,
    xmin = 10000,
    xmax = 110000,
    ymin = 0,
    ymax = 4000,
    scaled y ticks = true,
    cycle list name=yaf,
    yticklabel style={/pgf/number format/fixed},
    no markers,
    legend entries = {\alggreedyfastest, \algdenseg, \algdensec,\alglpstc},
    legend style={at={(0.5,1.0)}}
]
\addplot table [x=x, y=y1, col sep=comma] {time.csv};
\addplot table [x=x, y=y2, col sep=comma] {time.csv};
\addplot table [x=x, y=y3, col sep=comma] {time.csv};
\addplot table [x=x, y=y4, col sep=comma] {time.csv};
\pgfplotsextra{\yafdrawaxis{10000}{110000}{0}{4000}}
\end{axis}
\node[anchor=north east] at (0, -0.3) {(a)};
\end{tikzpicture}&
\begin{tikzpicture}
\begin{axis}[xlabel={$\abs{V(Z)}$}, ylabel= {time},
    width = 4.4cm,
    height = 3.8cm,
    xmin = 38000,
    xmax = 4800000,
    ymin = 0,
    ymax = 4000,
    scaled y ticks = true,
    cycle list name=yaf,
    yticklabel style={/pgf/number format/fixed},
    no markers,
    legend pos = north west
]
\addplot table [x=x1, y=y1, col sep=comma] {time.csv};
\addplot table [x=x1, y=y2, col sep=comma] {time.csv};
\addplot table [x=x1, y=y3, col sep=comma] {time.csv};
\addplot table [x=x1, y=y4, col sep=comma] {time.csv};
\pgfplotsextra{\yafdrawaxis{38000}{4800000}{0}{4000}}
\end{axis}
\node[anchor=north east] at (0, -0.3) {(b)};
\end{tikzpicture}

\end{tabular}
\end{center}
\end{subcaptiongroup}
\caption{Time in seconds as a function of the number of edges $\abs{E}$ and the number of wedges $\abs{V(Z)}$.
}
\label{fig:e-w-t}
\end{figure}

\textbf{Real-world datasets}:
We test our algorithms in publicly available, real-world datasets:
\dtname{Cora}~\cite{nr2015}\footnote{\url{https://networkrepository.com} \label{foot:nw-repo}} and \dtname{Cite-seer}~\cite{nr2015}\footref{foot:nw-repo} datasets are citation networks.
\dtname{Email-EU} is a collaboration network between researchers.\!\footnote{\url{https://toreopsahl.com/datasets/}}
\dtname{Facebook} is extracted from a friendship circle of Facebook.\!\footnote{\url{http://snap.stanford.edu}\label{foot:snap}}
\dtname{LastFM} is a social network 
of LastFM users.\!\footref{foot:snap} 
\dtname{PGP} is an interaction network of the users of the Pretty Good Privacy (PGP) algorithm.\!\footnote{\url{http://konect.cc/networks/arenas-pgp/}}
The details of the datasets are shown in Table~\ref{tab:real-data}.

\textbf{Results using real-world datasets}:
We present the results obtained using our algorithms in Table~\ref{tab:real-result}. We compare our algorithms in terms of scores, running time, and the percentage of the strong edges within
subgraphs they returned for a set of $\lambda$ values.
Since \algip invokes a sequence of integer programs, the algorithm does not scale for large datasets. 
We stopped the experiments that took over one hour. 
We always set $\epsilon = 0.01$ when testing each dataset with \algip.

First, let us compare the scores across the algorithms.
Our first observation is that \algip always yields the highest score with the tested datasets while all other algorithms perform similarly in terms of scores: in most cases, they produce approximately equal scores.
When \algip is not usable, \alglpstc has produced the maximum score except for $3$ outlier cases where \alggreedyfastest and \algdensec algorithms obtained the maximum score  $2$ times and $1$ time respectively. Nevertheless, for high $\lambda$s all of them produce less deviated scores when compared to lower $\lambda$s.
As expected, $\score{}$ increases as $\lambda$ increases for all $4$ algorithms. 

Next, let us look at column $\abs{E_s} \%$ which gives the percentages of strong edges in the returned subgraph.
Generally speaking, $\abs{E_s}\%$ monotonically decreases as $\lambda$ increases except for a few outlier cases. This is because when we assign a higher weight $\lambda$, it becomes more beneficial to include more weak edges.

Computation times are given in the $time$ columns of Table~\ref{tab:real-result}. 
\alggreedyfastest, \alglpstc, and \algdenseg run significantly slower than \algdensec.  
If we compare \alggreedyfastest and \algdenseg, for all the tested cases \algdenseg runs faster despite having to solve a sequence of minimum cuts. This is due to the implementation differences as \algdenseg uses a fast native library to compute the minimum cuts.

Despite \algip not being scalable for larger datasets, it runs faster than \alggreedyfastest except for four cases with the tested datasets.
For \dtname{Cora}, \dtname{Citeseer}, \dtname{PGP}, and \dtname{LastFM}  datasets, \alglpstc runs faster than all other algorithms except \algdensec.  However, for the two other remaining datasets, 
\alglpstc is the slowest in comparison to the other three algorithms.
We see that the running times are still reasonable in practice for the tested datasets; for example, we were able to compute the subgraph for the \dtname{Facebook}
dataset, with over $30\,000$ edges and $1\, 000 \,000$ wedges, in under ten minutes.

\begin{table}[ht!]
\caption{Co-authorship case-study for \dtname{DBLP} dataset with weighted variant of \algip. We set $\lambda = 0.8$ and $\epsilon = 0.01$. For each subgraph, we state the scores within brackets.
}
\label{tab:dblp-lam-ip-w}
\begin{tabular*}{\columnwidth} { l p{7.4cm}}
\toprule
S1 & P. S.Yu, C. C.Aggarwal, J.Han, W.Fan, J.Gao, X.Kong~(6.00)  \\
S2 & C. H. Q.Ding, F.Nie, H.Huang, D.Luo~(4.78) \\
S3 & S.Yan, J.Yan, N.Liu, Z.Chen, H.Xiong, Q.Yang, Y.Fu, Y.Ge, H.Zhu, E.Chen, C.Liu, Q.Liu, B.Zhang~(4.70)\\
S4 & S.Lin, H.Hsieh, C.Li~(4.23)\\
S5 & C.Faloutsos, J.Sun, S.Papadimitriou, H.Tong, L.Akoglu, T.Eliassi-Rad, B.Gallagher~(4.09)\\
S6 & Y.Liu, M.Zhang, S.Ma, L.Ru~(3.84)\\
S7 & H.Liu, J.Tang, X.Hu, H.Gao~(3.80)\\
S8 & D.Phung, S.Venkatesh, S. KumarGupta, S.Rana, S.Tsumoto, S.Hirano~ (3.76)\\
S9 & C.Böhm, I. S.Dhillon, C.Plant, C.Hsieh, P.Ravikumar~(3.54)\\
S10 & S.Günnemann, H.Kremer, T.Seidl, I.Assent, E.Müller, R.Krieger~(3.46) \\
\bottomrule
\end{tabular*}
\end{table}

\textbf{Case study}: Next, we conducted a case study for \dtname{DBLP}
~\cite{tang2008arnetminer}\footnote{\url{https://www.aminer.org/citation}} which contains co-authorship connections from top venues in data mining and machine learning~(SDM, NIPS, ICDM, KDD, ECMLPKDD, and WWW). Each node represents an author and each edge corresponds to a collaboration between two authors.
We removed the author pairs who have less than $3$ collaborations. The size of the dataset after prepossessing is $n = 4\,592$, $m =  5\,566$, and $\abs{Z(G)} =  26\,073$. To compute a marginal weight that corresponds to an author pair, we assign a weight for each paper as one divided by the number of authors. We then weigh each edge~(author-pair) by summing up the weights of all respective collaborations.
Then we ran a weighted version of \algip whose objective is to maximize the edge-weighted score,
\[
\frac{\sum_{\text{strong } e \in E(U)}w(e) + \lambda \sum_{\text{weak } e \in E(U)}w(e)}{\abs{U}}\quad. 
\]
We found top-10 non-overlapping subgraphs iteratively by deleting the returned subgraph in each iteration and then considering the remaining graph to find the next subgraph.
We set $\lambda = 0.8$ and $\epsilon = 0.01$. The list of author subgraphs is shown in Table~\ref{tab:dblp-lam-ip-w}. 
We see that the variant of \algip discovered subgraphs of prolific authors. 

\section{Concluding remarks}\label{sec:conclusions}

We introduced a novel dense subgraph discovery problem that takes into account the strength of ties within the subgraph.
Here we label each edge either as {strong} or {weak} based on the strong triadic closure principle (\stc). The \stc property requires that if one node {strongly} connects with two other nodes, then those two nodes should at least have a  {weak} connection between them. 
Our goal was to maximize a density-like measure defined as the sum of the number of strong edges and the number of weak edges weighted by a weight parameter, divided by the number of nodes within the subgraph.
We showed that our optimization problem is \np-hard, and connects the two well-known problems of finding dense subgraphs and maximum cliques. To solve the problem, we presented an exact algorithm based on integer linear programming. In addition, we presented a polynomial-time algorithm based on linear programming, a greedy heuristic algorithm, and two other straightforward algorithms based on the algorithms for the densest subgraph discovery.

The experiments with synthetic data showed that our approach recovers the
latent dense components. The experiments on real-world networks confirmed that our
proposed algorithms discovered the subgraphs reasonably fast in practice.
Finally, we presented a case study where our algorithm produced interpretable results suggesting the practical usefulness of our problem setting and algorithms.

The idea of combining the dense subgraph problem
with the \stc property opens up several lines of work. For example, 
instead of using the ratio between the number of edges and the number of nodes as the density measure, we can incorporate other density measures.

\begin{acks}
This research is supported by the \grantsponsor{⟨malsome⟩}{Academy of Finland project MALSOME}{} (\grantnum[]{malsome}{343045}).
\end{acks}

\bibliographystyle{ACM-Reference-Format}
\balance
\bibliography{bibliography}


\appendix
\section{Appendix}

\subsection{Computational complexity proofs}
\label{appendix:complexity}
\begin{proof}[Proof of Proposition~\ref{prop:np_hardness}]
We will show the \np-hardness of \prbstrwk by a reduction from the \np-hard \prbmaxclique problem. As $\lambda=0$, we simply write the score as $q(U, L)$ instead of $q(U, L; \lambda) = q(U, L; 0)$ for brevity. Assume that the set $U$ with the labeling $L$ is an optimal solution to the \prbstrwk problem with $\lambda = 0$, maximizing the score $q(U, L)$ while satisfying the \stc property. 
The density of the strong edges in $U$ is then

\begin{equation}
\label{equation:strong_density}
q(U, L) = \frac{\sum_{v \in U} \deg_s(v, U, L)}{2\abs{U}}.
\end{equation}

Consider a vertex $w$ in $U$ with the highest number $\deg_s(w, U, L)$ of strong edges connected to it. As the maximum number of strong edges, $\deg_s(w, U, L)$ has to be at least the average, 
\begin{equation}
\label{equation:average_degree}
\deg_s(w, U, L) \geq \frac{\sum_{v \in U} \deg_s(u, U, L)}{\abs{U}}.
\end{equation}

Then for any two vertices $v$ and $u$ that have strong edges $(v,w)$ and $(u,w)$ connecting them to $w$, there must be an edge between $v$ and $u$ to satisfy the \stc property.

Thus, the vertex $w$ and its strong neighbors form a clique $C$ with $\deg_s(w, U, L)+1$ vertices.
Consider only having these vertices in $C$ and having a labeling $L'$ that labels each edge in the clique as strong. This would satisfy the \stc property and would give a score of 
\begin{equation*}
q(C, L') = \frac{\abs{C}(\abs{C}-1)}{2\abs{C}} = \frac{\abs{C}-1}{2} = \frac{\deg_s(w, U, L)}{2} .
\end{equation*}

From Equations~\ref{equation:strong_density} and~\ref{equation:average_degree}, we get 
\begin{equation}
\label{equation:clique_score_comparison}
q(C, L') = \frac{\deg_s(w, U, L)}{2} \geq \frac{\sum_{v \in U} \deg_s(v, U, L)}{2\abs{U}} = q(U, L).
\end{equation}

Thus, the clique $C$ has at least the same score as $U$, which has the maximum score of all subgraphs, so $q(C, L') = q(U, L)$. This means that $C$ must be a maximum size clique in the input graph $G$, as larger cliques would give a higher score than $U$.

Therefore, by finding an optimal set of vertices $U$ and labeling $L$ we can find a maximum clique $C$. Thus, \prbstrwk is \np-hard.
\end{proof}

\begin{proof}[Proof of Proposition~\ref{prop:inapproximability}]
Assume that we can find a set $U$ with labeling $L$ that is an $\alpha$-approximation to \prbstrwk, while the optimal solution has value $q(C^*, L^*)$ with a maximum clique $C^*$ and labeling $L^*$. Consider then the vertex $w$ with the highest number of strong edges and construct the clique $C$ consisting of $w$ and its strong neighbors. Define the labeling $L'$ such that each edge in the clique $C$ is labeled as strong. Using Equation~\ref{equation:clique_score_comparison} and that $q(U, L)$ is an $\alpha$-approximation, we get 
\begin{equation*}
q(C, L') \geq q(U, L) \geq \alpha q(C^*, L^*).
\end{equation*}
But as the score of a clique $C$ with only strong edges is $q(C, L') = \frac{\abs{C}-1}{2}$, we have 

\begin{equation*}
\frac{\abs{C}-1}{2} \geq \alpha\frac{\abs{C^*}-1}{2}.
\end{equation*}

By solving for $\abs{C}$ and using $\alpha \leq 1$, we get

\begin{equation*}
\abs{C} \geq \alpha \pr{\abs{C^*}-1} + 1 \geq \alpha\abs{C^*}.
\end{equation*}

This means that we have an $\alpha$-approximation for \prbmaxclique.
Therefore, any inapproximability results for \prbmaxclique also apply for the $\lambda = 0$ case of \prbstrwk. Using the result by~\citet{zuckerman2006linear} then finishes the proof.
\end{proof}

To prove Proposition~\ref{prop:nphard2} we need the following lemma.

\begin{lemma}
\label{lem:opt}
Assume graph $G$. Let $X \subsetneq Y$ be two subgraphs, and let $L$ be a labeling defined on $Y$. Define
\[
    \Delta(X, Y) = \frac{m_s(Y) + \lambda m_w(Y) - m_s(X) - \lambda m_w(X)}{\abs{Y} - \abs{X}}\quad.
\]
If $\Delta(X, Y) < \score{Y, L}$, then $\score{Y, L} < \score{X, L}$, if $\Delta(X, Y) > \score{Y, L}$, then $\score{Y, L} > \score{X, L}$,
and if $\Delta(X, Y) = \score{Y, L}$, then $\score{Y, L} = \score{X, L}$.
\end{lemma}

\begin{proof}
Assume $\Delta(X, Y) < \score{Y, L}$. Multiply by $(\abs{Y} - \abs{X})\abs{Y}$ and subtract $\abs{Y} (m_s(Y)+\lambda m_w(Y))$ from both sides. Dividing by $-\abs{X}\abs{Y}$ then gives $\score{X,L} > \score{Y,L}$, proving the first claim. The proofs for other claims are identical.
\end{proof}

\begin{proof}[Proof of Proposition~\ref{prop:nphard2}]
We will prove the hardness by reducing an \np-hard problem \prbminSTC to our problem. In \prbminSTC, we are asked to label the full graph and minimize the number of weak edges~\citep{sintos2014using}.
Assume a graph $G$ with nodes $V = v_1, \ldots v_n$. We assume that $n \geq 5$.
We define a new graph $H$ that consists of $G$ and $k = \ceil{1/\lambda}(n + 1)/2$ cliques $C_i$ of size $n$. Let $\set{c_{ij}}$ be the nodes in $C_i$. For each $j$ and $i$, we connect $v_j$ with $c_{ij}$.

Let $U$ be the optimal subgraph of $H$ for \prbstrwk and $L$ be its labeling.
Let $L'$ be the labeling where every $E(C_i)$ is strong and the remaining edges are weak. Note that $\score{U} \geq \score{C_i, L'} = (n - 1)/2 \geq 2$ for any $C_i$. We claim that $U$ contains every node in $H$.


To prove the claim, let us define $W_i = U \cap C_i$. If $\abs{W_i} = 1, 2$, then $\Delta(U \setminus W_i, U) \leq 3/2 < \score{U}$ and Lemma~\ref{lem:opt} states that we can delete $W_i$ from $U$ and obtain a better score. Assume $3 \leq \abs{W_i} < n$. Let $c \in C_i \setminus W_i$. We can safely assume that the edges between $W_i$ and $G$ are weak; otherwise, we can relabel them as weak and compensate by labeling any weak edge in $W_i$ as strong. Now we can extend the labeling $L$ to $c$ by setting the edges from $c$ to $W_i$ as strong, and the possibly remaining edge as weak. We can show that $\Delta(U \setminus W_i, U) < \Delta(U,U \cup \set{c})$. Lemma~\ref{lem:opt}, applied twice, states that either deleting $W_i$ or adding $c$ improves the solution. Therefore, either $W_i = \emptyset$  or $W_i = C_i$.

Assume $W_i = C_i$ and $W_j = \emptyset$. The optimal labeling must be such that all edges between $W_i$ and $G$ are weak and the edges in $W_i$ are all strong. We can extend the same labeling scheme to $C_j$. Then $\Delta(U \setminus C_i, U) = \Delta(U, U \cup C_j)$. If $\Delta(U, U \cup C_j) > \score{U}$,  Lemma~\ref{lem:opt} implies that we improve the solution by adding $C_j$, which is a contradiction. Hence, $\Delta(U \setminus C_i, U) \leq \score{U}$. Lemma~\ref{lem:opt} implies that we can safely delete $C_i$. Applying this iteratively we arrive to an optimal solution with nodes only in $V$.
This cannot happen since then $\score{U} \leq (n - 1)/2$,
but then $\score{C_i \cup V, L'} = (n - 1)/2 + \lambda/2 > q(U)$. Therefore, $W_i = C_i$ for every $i$.

Finally, assume $v_j \notin U$. Then $\Delta(U, U \cup \set{v_j}) \geq \lambda k \geq (n + 1)/2 > (n - 1)/2 + \lambda \geq \Delta(U \setminus C_i, U)$. Lemma~\ref{lem:opt} states that either deleting any $C_i$ or adding $v_j$ improves the solution. This contradicts the optimality of $U$, so every $v_j \in U$.

Consequently, $V \subseteq U$. The optimal labeling must have every edge in $C_i$ as strong, the cross-edges between $C_i$ and $G$ as weak, and the labels for edges in $G$ solve \prbminSTC.
\end{proof}

\subsection{Proofs for Section~\ref{sec:ip}}
\label{appendix:ip}

\begin{proof}[Proof of Proposition~\ref{prop:frac}]
Let us write $f(U, L) = m_s(U, L) + \lambda m_w(U, L)$.
Note that
\begin{equation}
\label{eq:alphapos}
    f(U(\alpha), L(\alpha)) - \alpha \abs{U(\alpha)} \geq 0.
\end{equation}
Assume $\alpha > \alpha^*$. If $U(\alpha) \neq \emptyset$, then Eq.~\ref{eq:alphapos} implies that
\[
    \score{U(\alpha), L(\alpha)} \geq \alpha > \alpha^*,
\]
which contradicts the optimality of $\alpha^*$. Thus, $U(\alpha) = \emptyset$.

Assume $\alpha < \alpha^*$. Then
\[
\begin{split}
    f(U(\alpha), L(\alpha)) - \alpha \abs{U(\alpha)} & \geq f(U^*, L^*) - \alpha \abs{U^*} \\
    & > f(U^*, L^*) - \alpha^* \abs{U^*} = 0.
\end{split}
\]
That is, $f(U(\alpha), L(\alpha)) >  \alpha \abs{U(\alpha)}$, implying in turn that $U(\alpha) \neq \emptyset$ and $\score{U(\alpha), L(\alpha)} > \alpha$.
\end{proof}

\begin{proof}[Proof of Proposition~\ref{prop:opt-ip-approx}]
Let $L$ and $U$ be the values of the interval when binary search is terminated. Note that $\alpha \geq L$ due to Proposition~\ref{prop:frac}.
We know that $U - L \leq \epsilon L$ and $L \leq \alpha^* \leq U$. Thus,
$\alpha^* - L \leq U - L \leq \epsilon L$, or $\alpha^*  \leq (1 + \epsilon) L \leq (1 + \epsilon) \alpha$.
\end{proof}

\begin{proof}[Proof of Proposition~\ref{prop:opt-ip-exact}]
Let $\alpha$ be the score of the solution $X$, $L$ returned by $\algip$, and let $\alpha^*$ be the score of the optimal solution $X^*$, $L^*$ for \prbstrwk.
We will show that if $\alpha < \alpha^*$, then $\alpha^* - \alpha \geq 1/(bn^2)$, which contradicts with the fact that $\alpha^* - \alpha \leq \epsilon \alpha < \epsilon n / 2 = 1/(bn^2)$.

To prove the claim, let $\Delta = \alpha^* - \alpha$. Then
\[
\begin{split}
	\Delta & =
	\frac{m_s(X^*) + \frac{a}{b} m_w(X^*)}{\abs{X^*}} - \frac{m_s(X) + \frac{a}{b} m_w(X)}{\abs{X} } \\
	& = \frac{ \abs{X} (bm_s(X^*) + a m_w(X^*)) -  \abs{X^*}(bm_s(X) + a m_w(X))}{b \abs{X} \abs{X^*}}.
\end{split}
\]
Note that the numerator and the denominator are both integers.
Consequently, if $\Delta > 0$, then $\Delta \geq 1/(bn^2)$.
It follows that if we set $\epsilon = \frac{2}{bn^3}$, then \algip finds the optimal solution in  $\bigO{\log n + \log b}$ number of rounds.
\end{proof}

\subsection{Proofs for Section~\ref{sec:lp}}
\label{appendix:lp}

\begin{proof}[Proof of Proposition~\ref{prop:frac2}]
Scaling $(x^*, y^*, z^*)$ by any constant $c > 0$ does not
change the value of $r(\cdot, \cdot, \cdot)$ nor does it change the validity of the constraints in Eqs.~\ref{ip_con_1}--\ref{ip_con_3}. Therefore, we can safely assume that $x_e^*, z_e^* \leq 1$ and $y_i^* \leq 1$ and $i \in V$, for any $e \in E$ and $i \in V$.
The claim now follows by repeating the steps of the proof of Proposition~\ref{prop:frac}.
\end{proof}

\begin{proof}[Proof of Proposition~\ref{prop:lp}]
Scaling $(x^*, y^*, z^*)$ by any constant $c > 0$ does not
change the value of $r(\cdot, \cdot, \cdot)$ nor does it change the validity of the constraints in Eqs.~\ref{ip_con_1}--\ref{ip_con_3}. Therefore, we can safely require that $\sum y_i = 1$, which immediately proves the claim.
\end{proof}

\end{document}